\documentclass[a4paper,UKenglish]{article}

\pdfoutput=1 

\usepackage[utf8]{inputenc}
\usepackage{a4wide}
\usepackage{xcolor}
\usepackage{graphicx}
\usepackage{xspace}
\usepackage{url}
\usepackage{amsmath,amssymb,bm,amsfonts,amsthm}
\usepackage{hyperref}
\usepackage{authblk}
\usepackage{adjustbox}
\usepackage{bm}
\usepackage{multirow}
\usepackage{tikz}
\usepackage{tikz-network}
\usepackage{booktabs}

\DeclareMathOperator{\dep}{dep}
\DeclareMathOperator{\arr}{arr}

\let\epsilon\varepsilon

\usepackage[linesnumbered,ruled,noline,boxed]{algorithm2e}
\newlength{\commentWidth}
\let\oldnl\nl
\newcommand{\nonl}{\renewcommand{\nl}{\let\nl\oldnl}}
\DontPrintSemicolon
\SetKwInOut{Input}{input}
\SetKwInOut{Output}{output}
\SetInd{0.2em}{0.5em}
\SetKw{Return}{return}
\SetKw{OurIf}{if}
\SetKw{OurThen}{then}
\SetKwIF{If}{ElseIf}{Else}{if}{then}{else if}{else}{}
\SetKwFor{For}{for}{do}{}
\SetKwFor{for}{for}{}{}
\SetKwFor{While}{while}{do}{}
\SetKwFor{Repeat}{Repeat}{}{}
\SetKwFor{ForEach}{foreach}{do}{}
\SetKwRepeat{Do}{Do}{while}
\SetKwIF{Do}{notused}{WhileOfDo}{Do}{}{notused}{while}{done}
\SetKwFunction{Finalize}{\mbox{Finalize}}
\SetKwComment{Comment}{/* }{ */}
\SetKw{feach}{foreach}
\usepackage{tikz}
\usetikzlibrary{decorations.pathreplacing}
\usetikzlibrary{plotmarks}
\usetikzlibrary{positioning,automata,arrows}
\usetikzlibrary{shapes.geometric}
\usetikzlibrary{decorations.markings}
\usetikzlibrary{positioning, shapes, arrows}

\PassOptionsToPackage{usenames,dvipsnames,svgnames}{xcolor}

\newtheorem{theorem}{Theorem}
\newtheorem{lemma}[theorem]{Lemma}

\newtheorem{fact}[theorem]{Fact}

\newcommand{\bcl}{\mbox{$\mathrm{\textsc{Fast}}$}}
\newcommand{\bmnr}{\mbox{$\mathrm{\textsc{bmnr}}$}}
\newcommand{\onbra}{\mbox{\textsc{Onbra}}}

\newcommand{\tcost}{\mbox{$\Theta$}}
\newcommand{\cost}{\mbox{$\Gamma$}}
\newcommand{\tcf}{\mbox{$\mathtt{TC}$}}

\newcommand{\fa}{{Fa}}
\newcommand{\fo}{{Fo}}
\newcommand{\sh}{{Sh}}
\newcommand{\sfa}{{SFa}}
\newcommand{\sfo}{{SFo}}

\newcommand{\iarr}{\overline}
\newcommand{\idep}{\underline}

\DeclareMathOperator{\processcosts}{\mathtt{Finalize}}

\newcommand{\intervs}[1]{{\mathcal{I}}_{#1}}

\usepackage{bm}

\title{Making Temporal Betweenness Computation Faster and Restless}

\author[1]{Filippo Brunelli}
\affil[1]{European Commission –- JRC, Seville, Spain}
\author[2]{Pierluigi Crescenzi}
\affil[2]{Gran Sasso Science Institute, L'Aquila, Italy}
\author[3]{Laurent Viennot}
\affil[3]{Inria, DI ENS, Paris, France}

\begin{document}

\maketitle

\begin{abstract}
Bu\ss\ et al [KDD 2020] recently proved that the problem of computing the betweenness of all nodes of a temporal graph is computationally hard in the case of foremost and fastest paths, while it is solvable in time $O(n^3T^2)$ in the case of shortest and shortest foremost paths, where $n$ is the number of nodes and $T$ is the number of distinct time steps.  A new algorithm for temporal betweenness computation is introduced in this paper. In the case of shortest and shortest foremost paths, it requires $O(n + M)$ space and runs in time $O(nM)=O(n^3T)$, where $M$ is the number of temporal edges, thus significantly improving the algorithm of Bu\ss\ et al in terms of time complexity (note that $T$ is usually large). Experimental evidence is provided that our algorithm performs between twice and almost 250 times better than the algorithm of Bu\ss\ et al. Moreover, we were able to compute the exact temporal betweenness values of several large temporal graphs with over a million of temporal edges. For such size, only approximate computation was possible by using the algorithm of Santoro and Sarpe [WWW 2022]. Maybe more importantly, our algorithm extends to the case of \textit{restless} walks (that is, walks with waiting constraints in each node), thus providing a polynomial-time algorithm (with complexity $O(nM)$) for computing the temporal betweenness in the case of several different optimality criteria. Such restless computation was known only for the shortest criterion (Rymar et al [JGAA 2023]), with complexity $O(n^2MT^2)$. We performed an extensive experimental validation by comparing different waiting constraints and different optimisation criteria. Moreover, as a case study, we investigate six public transit networks including Berlin, Rome, and Paris. Overall we find a general consistency between the different variants of betweenness centrality. However, we do measure a sensible influence of waiting constraints, and note some cases of low correlation for certain pairs of criteria in some networks.
\end{abstract}




\section{Introduction}
\label{sec:introduction}

Social network analysis is usually considered to start with the book of Moreno~\cite{Moreno1934}, where \textit{sociograms}  (that is, graphs) are used to study the relationships between kids from kindergarten to the 8th grade. Successively, Bavelas~\cite{Bavelas50} used sociogram analysis (that is, graph mining) techniques to identify the most important members of a group. Several notions of node importance were successively introduced until Freeman~\cite{Freeman77} proposed precise different definitions of \textit{node centrality}, such as the degree centrality, the closeness centrality, and the \textit{betweenness centrality}. This latter centrality, which measures how often a node participates in an optimal path, has been repeatedly applied in different research contexts, such as, for example, the analysis of brain, collaboration, citation and, in general, social networks. From a computational point of view, Brandes~\cite{Brandes_2001} proposed an algorithm for computing the betweenness centrality of all nodes in time $O(nm)$, where $n$ is the number of nodes and $m$ is the number of edges. Due to the huge size of some real-world networks, several approximation algorithms for computing the betweenness centrality have also been proposed in the last twenty years (such as, for example, the ABRA and KADABRA algorithms described in~\cite{RiondatoU18,Borassi2019}), mostly based on sampling techniques.

More recently, the notion of betweenness centrality has been also applied to the case of temporal graphs. Indeed, many real-world complex networks evolve over time, in the sense that edges can appear and disappear at specific time instants. As observed in~\cite{Latapy_2018}, this is due to the fact that interactions between nodes take place over time, as it happens, for example, in the case of collaboration, communication, user-product, and transport networks. Many different modelizations of these evolving networks have been proposed in the literature, such as the \textit{evolving graph} model analysed in~\cite{Ferreira2004}, the \textit{time-dependent graph} model studied in~\cite{Foschini2014}, the time-varying graph model~\cite{Casteigts2012}, or the \textit{link stream} model introduced in~\cite{Latapy_2018}, 
and the temporal graph model~\cite{Michail2016}. A \textit{temporal graph} is a collection of \textit{temporal edges} over a fixed set of nodes. Each temporal edge is an edge $(u,v)$ with an associated \textit{availability} or \textit{departure} time $\tau$ and a \textit{traversal} or \textit{travel} time $\lambda$. A temporal edge $e=(u,v,\tau,\lambda)$ can then be traversed starting from $u$ at time $\tau$ and arriving in $v$ at time $\tau+\lambda$. The notion of walk (or of path, if node repetitions are not allowed) can be easily adapted to the case of temporal graphs by imposing the natural condition of traversing the temporal edges in ascending time: that is, for any two consecutive temporal edges $(u,v,\tau_{1},\lambda_{1})$ and $(v,w,\tau_{2},\lambda_{2})$ in the walk, $\tau_{2}\geq\tau_{1}+\lambda_{1}$.\footnote{In this paper, we assume that the traversal time is always positive, which corresponds to the case called \textit{strict} in the literature.} However, introducing the temporal dimension arises different notions of optimal walks or paths, the most common used being the \textit{shortest one} (with the fewest number of temporal edges), the \textit{foremost one} (with the earliest arrival time), and the fastest one (with the smallest total travel time). Once a notion of optimal walk or path is adopted, the corresponding notion of betweenness can be analysed in terms of its computational complexity. Indeed, this has been done in~\cite{BussMNR20} in the case of paths, where the authors proved that computing the shortest betweenness of all nodes can be done in time $O(n^{3}T^{2})$, where $n$ is the number of nodes and $T$ is the number of distinct time steps in the temporal graph, while computing the foremost and the fastest betweenness is \#P-hard, that is, most likely computationally intractable. These cases are also studied in~\cite{Rymar2023}, where, among other results, the authors proved that the restless shortest betweenness (in the case of walks) can be computed in time $O(n^{2}MT^{2})$, where $M$ is the number of temporal edges (interestingly, the \#P-hardness result does not hold when considering walks rather than paths). A \textit{restless} walk (or path) is a walk in which we do not ``wait'' at the same node more than $\beta$ time units, where $\beta$ is a constant: that is, for any two consecutive temporal edges $(u,v,\tau_{1},\lambda_{1})$ and $(v,w,\tau_{2},\lambda_{2})$ in the walk, $\tau_{2}\leq\tau_{1}+\lambda_{1}+\beta$ in addition to $\tau_{2}\geq\tau_{1}+\lambda_{1}$. Note that, in the restless case, even deciding whether there exists a path between two specific nodes is NP-complete~\cite{Casteigts2021}. Hence, in this case, we are forced to consider walks instead of paths. Our main contributions are the following.

\begin{itemize}
\item We propose an algorithm for computing the betweenness in time $O(nM)=O(n^{3}T)$ in the case of shortest and shortest foremost paths (note that, in these cases, optimal non-restless walks are always paths). Our algorithm significantly outperforms the previously known algorithm (whose complexity was $O(n^{3}T^{2})$) for computing the shortest and the shortest foremost betweenness~\cite{BussMNR20}.

\item We propose an algorithm for computing the betweenness in time $O(nM)$ in the case of fastest, foremost, shortest, shortest fastest, and shortest foremost restless walks. The only previously known polynomial-time algorithm (with complexity $O(n^{2}MT^{2})$) was in the case of shortest restless walks~\cite{Rymar2023}: our algorithm significantly improves this algorithm. As far as we know, in all the other cases our algorithm is the first polynomial-time one and we conjecture it is optimal. Indeed, the algorithm is based on a new way of counting the number of optimal walks from a given source in time $O(M)$. This complexity is clearly optimal.

\item We perform an extensive experimental evaluation of our non-restless algorithm, based on a diverse set of real-world networks that includes all publicly available networks from the works of~\cite{BussMNR20,Santoro_2022,Becker2023}. In particular, we compare the execution time of our algorithm and of the algorithm proposed in~\cite{BussMNR20} for computing the shortest and the shortest foremost betweenness. It turns out that our algorithm is between twice and almost 250 faster than the algorithm of~\cite{BussMNR20}.

\item By using our non-restless algorithm, we are able to compute in a reasonable amount of time the shortest betweenness of all nodes of three quite large temporal graphs analysed in~\cite{Santoro_2022}, for which only approximate values were available so far by making use of the ONBRA approximation algorithm proposed in that paper and based on a sampling technique. From the results reported in~\cite{Santoro_2022}, it also follows that our algorithm is almost always significantly faster than ONBRA.

\item We apply our restless algorithm for computing the fastest, foremost, shortest, shortest foremost, and shortest fastest betweenness of all nodes of the temporal graphs considered in the first experiment. By referring to the (weighted) Kendall's $\tau$ correlation and to the intersection of the top-50 node sets, we compare the node rankings produced by the different betweenness measures and by different waiting constraints and we observe that there exists a general consistency between the different variants of betweenness centrality. We do also measure a sensible influence of waiting constraints, and note some cases of low correlation for certain pairs of criteria in some networks.

\item As a case study, we apply our algorithm to the analysis of several public transport networks among the ones published in~\cite{Kujala2018,Crescenzi2019}. In particular, we compare the execution time, the (weighted) Kendall's $\tau$, and the size of the intersection of the top-100 node sets for two different notions betweenness, that is, the shortest fastest and the shortest foremost betweenness. We observe a strong consistency between the two different variants of betweenness centrality, and a much lower consistency with the rankings produced by the betweenness of the static underlying graph, thus suggesting that this latter measure cannot be used as a `proxy' of the shortest fastest and the shortest foremost betweenness.
\end{itemize}

Both our algorithm and the algorithm proposed in~\cite{BussMNR20} follow a two phase approach (that is a path counting forward phase and a betweenness accumulation backward phase) and are both inspired by Brandes' algorithm~\cite{Brandes_2001}. A first difference between our algorithm and the algorithm of~\cite{BussMNR20} is that the latter focuses on temporal vertices and explore their temporal neighbors, while our algorithm focuses on temporal edges and explore their temporal extensions. A second (and maybe, main) difference is that we leverage on two orderings of the temporal edges to overall consider each temporal edge (both in the forward phase and in the backward phase) a constant number of times, rather than considering a temporal neighbor for each of its predecessors. Third, our data structure allows to store predecessors in linear space with respect to the number of temporal edges. Finally, we note that the approach used in~\cite{Santoro_2022} is different as it is an approximation algorithm and it works through sampling. However, the way the paths are counted is similar to~\cite{BussMNR20}.
\section{Related work}

The literature on centrality measures being vast (as demonstrated by the clever periodic table of network centrality developed by David Schoch~\cite{Schoch2017,PeriodicWebSite}), we restrict our attention to approaches that are closest to ours, that is, to the realm of temporal graphs. Several introductions to temporal graphs also include surveys on temporal centrality measures (see, e.g., \cite{Holme_2015,Latapy_2018,Santoro_2011}). Clearly related to our work is the literature on the efficient computation of temporal paths and walks, such as the seminal paper of Bui-Xuan, Ferreira, and Jarry~\cite{XuanFJ03} and the more recent paper by Wu et al~\cite{Wu_2016} (the reader may also refer to the quite exhaustive analysis of this literature appeared in~\cite{BrunelliV2022}). Besides the references given in the introduction, our paper is mostly related to all work on the definition and computation of different temporal centrality measures, such as (in order of appearance) the temporal pagerank defined in~\cite{Rozenshtein_2016}, the temporal Katz centrality introduced in~\cite{Beres_2018}, the temporal reachability used in~\cite{Falzon2018}, the~$f$-PageRank centrality defined in~\cite{Lv2019}, the temporal betweenness centrality defined in~\cite{Tsalouchidou_2020}, the temporal closeness centrality treated in~\cite{Crescenzi2020,Oettershagen_2020}, the temporal walk centrality introduced in~\cite{Oettershagen_2022}, and the temporal betweenness centrality analysed in~\cite{Simard_2023}, just to mention the most recent ones. Finally, more ``local'' notions of centrality in temporal graphs have also been analysed such as the temporal version of ego betweenness introduced in Ghanem~\cite{Ghanem_2017} and the pass-through degree defined in~\cite{Becker2023}: these centralities are clearly more efficient in terms of execution time, but not always satisfying in terms of the quality of their rankings.

\section{Basic definitions and results} 
\label{sec:preliminaries}

\noindent\textbf{Temporal graphs.} A \textit{temporal graph} is a tuple ${G}=(V,{E},\beta)$, where $V$ is the set of \textit{nodes}, $E$ is the set of temporal edges, and $\beta\in \mathbb{N}\cup\{+\infty\}$ is the maximum waiting-time (we say that \textit{waiting is unrestricted} when $\beta=+\infty$). A \textit{temporal edge} $e$ is a quadruple $(u,v,\tau,\lambda)$, where $u\in V$ is the \textit{tail} of $e$, $v\in V$ is the \textit{head} of $e$, $\tau\in\mathbb{N}$ is the \textit{departure} (or \textit{availability}) \textit{time} of $e$, and $\lambda\in\mathbb{N}^{+}$ is the \textit{travel} (or \textit{traversal}) \textit{time} of $e$. We also define the \textit{arrival time} of $e$ as $\tau+\lambda$, and we let $\dep(e)=\tau$ and $\arr(e)=\tau+\lambda$ denote the departure time and arrival time of $e$, respectively. We let $n=|V|$ and $M=|E|$ denote the number of nodes and temporal edges, respectively, and $T$ denote the number of distinct availability times. Finally, for any node $v$, $E_{v}^\mathrm{h}$ will denote the set of temporal edges whose head is $v$.

\noindent\textbf{Temporal graph representation.} We use a \textit{doubly-sorted representation} of a temporal graph $(V,E,\beta)$, which consists of two lists $E^{\arr}$ and $E^{\dep}$, each containing $|E|$ quadruples representing the temporal edges in $E$: $E^{\arr}$ is a list sorted by non-decreasing arrival time and $E^{\dep}$ is a list sorted by non-decreasing departure time. More precisely, we assume that $E^{\dep}$ is specified through implicit pointers from $E^{\dep}$ to $E^{\arr}$, that link each (logical) quadruple in $E^{\dep}$ to the (physical) quadruple in $E^{\arr}$ representing the same temporal edge.

\noindent\textbf{Temporal walks.} Given a temporal graph ${G}=(V,{E},\beta)$, a \textit{walk} $W$ from (a source) $s$ to (a target) $t$, or a \textit{$st$-walk} for short, is a sequence of temporal edges $e_{i}=(u_{i},v_{i},\tau_{i},\lambda_{i})$ for $i\in[k]$, such that $s=u_1$, $v_k=t$, and, for each $i\in[k-1]$,\footnote{In the following, for any non-negative integer $n$, $[n]$ will denote the set $\{1,2,\ldots,n\}$, with $[0]=\emptyset$.} $u_{i+1}=v_{i}$ and $\arr(e_{i}) \le \tau_{i+1} \le \arr(e_{i}) + \beta$ ($W$ is also called a $se_{k}$-walk). A walk is said to be a \textit{path} if, for any $i,j\in [k]$ with $i\neq j$, $u_{i}\neq u_{j}$ and $u_{i}\neq v_{k}$.  Note that, since travel times are positive, walks are \textit{strict} in the sense that $\tau_{i} < \tau_{i+1}$, for $i\in[k-1]$. The \textit{departure time} $\dep(W)$ of $W$ is defined as $\dep(e_{1})$, while the \textit{arrival time} $\arr(W)$ of $W$ is defined as $\arr(e_{k})$. The \textit{duration} of $W$ is defined as $\arr(W)-\dep(W)$. We say that a temporal edge $e=(t,w,\tau,\lambda)$ \textit{extends} $W$ if $\arr(W)\le \tau\le \arr(W)+\beta$. When $e$ extends $W$, we can indeed define the $sw$-walk $W.e=\langle e_1,\ldots, e_k,e\rangle$ from $s$ to $w$. Moreover, we also say that $e$ extends $e_k$ as it indeed extends any $se_{k}$-walk. Finally, we say that a temporal edge $e$ is \textit{$s$-reachable} when there exists a $se$-walk.

Given a temporal graph ${G}=(V,{E},\beta)$, a $st$-walk $W$ is a \textit{shortest} (respectively, \textit{foremost}, \textit{latest}, and \textit{fastest}) \textit{walk}, if there is no $st$-walk that contains less temporal edges than $W$ (respectively, has an earlier arrival time, has a later departure time, and has a smaller duration). Moreover, a \textit{shortest foremost} (respectively, \textit{latest} and \textit{fastest}) \textit{walk} is a foremost (respectively, latest and fastest) temporal walk that is not longer than any other foremost (respectively, latest and fastest) temporal walk. In the following we will focus on shortest (\sh) and shortest foremost (\sfo) walks, since these walks will allow us to introduce our algorithms in an easier way, without hiding the generality of our approach. In the appendix, we show how our algorithms can be adapted to the other types of walks by introducing the notions of cost and target cost structure (for the sake of brevity, all the proofs are included in Appendix~\ref{sec:proofs1}).

\noindent\textbf{Two basic facts of \sh\ and \sfo\ walks.} The following two results are easy to be proved, since they immediately follows from the definition of \sh\ and \sfo\ walks.

\begin{fact}\label{fact:sfo_prefixoptimality}
Let $G=(V,E,\beta)$ be a temporal graph. For any node $s\in V$, if a walk $W$ with last temporal edge $f\in E$ is a \sh\ one among the $sf$-walks, and $e\in E$ is a temporal edge of $W$, then the prefix of $W$ up to the temporal edge $e$ is a \sh\ walk among the $se$-walks.
\end{fact}

\begin{fact}\label{fact:sfo_toptimalimpliescoptimal}
Given a temporal graph ${G}=(V,{E},\beta)$, let $W$ be a \sfo\ $st$-walk (for some $s,t\in V$) and let $e$ be the last temporal edge of $W$. Then, $W$ is a \sh\ walk among the $se$-walks.
\end{fact}

\noindent\textbf{\sfo\ betweenness.} Given a temporal graph ${G}=(V,{E},\beta)$, two nodes $s,t\in V$ with $s\neq t$, and a temporal edge $e\in E$, we let $\sigma^{*}_{s,e,t}$ denote the number of \sfo\ walks from $s$ to $t$ that contain $e$. We also denote by $\sigma^{*}_{s,t}=\sum_{e\in E_{t}^\mathrm{h}}\sigma^{*}_{s,e,t}$ the number of \sfo\ walks from $s$ to $t$.

We define the $s$-\sfo\ betweenness of a temporal edge $e$ as $b_{s,e} = \sum_{v\in V:\chi_{s,t}=1}\sigma^{*}_{s,e,t}/\sigma^{*}_{s,t}$, where $\chi_{s,t} = 1$ if $s \neq t$ and there exists a $st$-walk (and, hence, $\sigma^{*}_{s,t}\neq0$), and $\chi_{s,t} = 0$ otherwise. 

Given three pairwise-distinct nodes $s$, $u$, and $t$ we denote by $\sigma^{*}_{s,u,t}=\sum_{e\in E_{u}^\mathrm{h}}\sigma^{*}_{s,e,t}$ the number of \sfo\ walks from $s$ to $t$ that contain $u$. Note that a walk where $u$ appears $\mu$ times is counted with multiplicity $\mu$.\footnote{It is quite natural to take into account how many times a node appears in a walk when considering walks rather than paths in the betweenness definition.} The \sfo\ betweenness of a vertex $u$ is defined as $b_{u} = \sum_{s,t\in V\setminus\{u\}:\chi_{s,t}=1}\sigma^{*}_{s,u,t}/\sigma^{*}_{s,t}$. The \sfo\ betweenness of any vertex $u$ can easily be computed from the $s$-\sfo\ betweenness $b_{s,e}$ of all temporal edges $e$ entering $u$.

\begin{fact}\label{fact:betweennessformula}
Given a temporal graph ${G}=(V,{E},\beta)$ and a node $u\in V$, the following holds: $b_u = \sum_{s\in V\setminus\{u\}}\left(\sum_{e\in E^\mathrm{h}_u}b_{s,e}-\chi_{s,u}\right)$.
\end{fact}

Given a temporal graph ${G}=(V,{E},\beta)$ and a temporal edge $e\in E$, let $\sigma_{s,e}$ (respectively, $\sigma^{*}_{s,e}$) denote the number of \sh\ (respectively, \sfo) $se$-walks, where a $se$-walk $W$ is \sfo\ if there is no $se$-walk $X$ such that $(\arr(X),|X|)\vartriangleleft (\arr(W),|W|)$ (in the following, for any $a,b,c,d\in\mathbf{N}$, $(a,b)\vartriangleleft(c,d)$ if and only $(a<c)\vee(a=c\wedge b<d)$). Moreover, given two nodes $s,t\in V$ with $s\neq t$, let $\mathcal{W}_{s,e,t}$ denote the set of \sfo\ $st$-walk containing $e$ (hence, $|\mathcal{W}_{s,e,t}|=\sigma^*_{s,e,t}$). Each walk $W\in \mathcal{W}_{s,e,t}$ can be decomposed into a prefix $W_{1}$ (from $s$ to $v$) ending with $e$ and a suffix $W_{2}$ (from $v$ to $t$). Let $\theta_{s,e,t}$ denote the number of distinct suffixes of walks in $\mathcal{W}_{s,e,t}$ (eventually including the empty suffix). 

\begin{fact}\label{fact:sfo_prefixtimessuffix}
Given a temporal graph ${G}=(V,{E},\beta)$, two nodes $s,t\in V$ with $s\neq t$, and a temporal edge $e=(u,v,\tau,\lambda)$, the following hold: $\sigma^*_{s,e,t}=\sigma_{s,e}\cdot\theta_{s,e,t}$.
\end{fact}

\paragraph*{Successors, predecessors, and edge betweenness recursive formulation.}
Given a temporal graph ${G}=(V,{E},\beta)$, two nodes $s,t\in V$ with $s\neq t$, and a temporal edge $e\in E$, let $\mathtt{succ}_{s,e,t}$ denote the set of temporal edges $f$ such that $e$ and $f$ are one after the other in a walk $W\in\mathcal{W}_{s,e,t}$. Moreover, let $\mathtt{succ}_{s,e}=\bigcup_{t\in V\setminus\{s\}}\mathtt{succ}_{s,e,t}$. In the following, if $f\in\mathtt{succ}_{s,e}$, we say that $f$ is \textit{successor} of $e$ and that $e$ is a \textit{predecessor} of $f$.

\begin{lemma}
\label{lem:backwardlemma}
Given a temporal graph ${G}=(V,{E},\beta)$ and a temporal edge $e=(u,v,\tau,\lambda)$, the following hold:
\[
b_{s,e} = \sigma_{s,e}\sum_{f\in\mathtt{succ}_{s,e}}\frac{b_{s,f}}{\sigma_{s,f}}+\left\{\begin{array}{ll}\frac{\sigma^{*}_{s,e}}{\sigma^{*}_{s,v}} & \mbox{if $\sigma^{*}_{s,e}>0$,}\\
0 & \mbox{otherwise.}\end{array}\right.   
\]
\end{lemma}

The above lemma will be used in the backward phase of our algorithms in order to compute the $s$-\sfo\ betweenness of a temporal edge $e$ having already computed the $s$-\sfo\ betweenness of all successors $f$ of $e$. In Appendix~\ref{sec:examplecentrality} we show an example of a temporal graph and of the centrality values of its nodes with respect to the different notions of optimal walks.

\begin{figure*}[ht]
    \centering
    \begin{adjustbox}{width=\textwidth}
        \SetVertexStyle[FillColor=white]
\SetEdgeStyle[Color=black]
\begin{tikzpicture}[x=2cm,y=2cm]
\Vertex[x=0,y=0,label=$v_{1}$]{v1}
\Vertex[x=2,y=0,label=$v_{2}$]{v2}
\Vertex[x=1,y=2.5,label=$v_{3}$]{v3}
\Vertex[x=3,y=2.5,label=$v_{4}$]{v4}
\Vertex[x=4,y=0,label=$v_{5}$]{v5}
\Vertex[x=6,y=0,label=$v_{6}$]{v6}
\Vertex[x=3,y=-2.5,label=$v_{7}$]{v7}
\Vertex[x=5,y=-2.5,label=$v_{8}$]{v8}
\Edge[Direct,label={$\stackrel{\iarr{1}}{1,1}$},bend=30](v1)(v2)
\Edge[Direct,label={$\stackrel{\iarr{3}}{4,1}$},bend=-30](v1)(v2)
\Edge[Direct,label={$\stackrel{\iarr{7}}{2,5}$},bend=-50](v1)(v5)
\Edge[Direct,label={$\stackrel{\iarr{5}}{4,2}$},bend=30](v2)(v5)
\Edge[Direct,label={$\stackrel{\iarr{9}}{6,3}$},bend=-30](v2)(v5)
\Edge[Direct,label={$\stackrel{\iarr{2}}{2,2}$}](v2)(v3)
\Edge[Direct,label={$\stackrel{\iarr{4}}{4,1}$}](v3)(v4)
\Edge[Direct,label={$\stackrel{\iarr{6}}{5,1}$}](v4)(v2)
\Edge[Direct,label={$\stackrel{\iarr{8}}{7,1}$}](v5)(v7)
\Edge[Direct,label={$\stackrel{\iarr{12}}{9,1}$}](v7)(v8)
\Edge[Direct,label={$\stackrel{\iarr{13}}{10,1}$}](v5)(v8)
\Edge[Direct,label={$\stackrel{\iarr{10}}{6,3}$},bend=30](v5)(v6)
\Edge[Direct,label={$\stackrel{\iarr{11}}{9,1}$},bend=-30](v5)(v6)
\node at (4.25,0) {
\begin{tabular}{l}
$E^{\mathrm{arr}}$\\
$\iarr{1}: (v_{1},v_{2},1,1)$\\
$\iarr{2}: (v_{2},v_{3},2,2)$\\
$\iarr{3}: (v_{1},v_{2},4,1)$\\
$\iarr{4}: (v_{3},v_{4},4,1)$\\
$\iarr{5}: (v_{2},v_{5},4,2)$\\
$\iarr{6}: (v_{4},v_{2},5,1)$\\
$\iarr{7}: (v_{1},v_{5},2,5)$\\
$\iarr{8}: (v_{5},v_{7},7,1)$\\
$\iarr{9}: (v_{2},v_{5},6,3)$\\
$\iarr{10}: (v_{5},v_{6},6,3)$\\
$\iarr{11}: (v_{5},v_{6},9,1)$\\
$\iarr{12}: (v_{7},v_{8},9,1)$\\
$\iarr{13}: (v_{5},v_{8},10,1)$\\
\end{tabular}
};
\node at (5.5,0) {
\begin{tabular}{l}
$E^{\mathrm{dep}}$\\
$1: \iarr{1}$\\
$2: \iarr{2}$\\
$3: \iarr{7}$\\
$4: \iarr{3}$\\
$5: \iarr{4}$\\
$6: \iarr{5}$\\
$7: \iarr{6}$\\
$8: \iarr{9}$\\
$9: \iarr{10}$\\
$10: \iarr{8}$\\
$11: \iarr{11}$\\
$12: \iarr{12}$\\
$13: \iarr{13}$\\
\end{tabular}
};
\node at (6.75,0.6) {
\begin{tabular}{l}
$E_\mathrm{node}^{\mathrm{dep}}$\\
$1: [\iarr{1},\iarr{7},\iarr{3}]$\\
$2: [\iarr{2},\iarr{5},\iarr{9}]$\\
$3: [\iarr{4}]$\\
$4: [\iarr{6}]$\\
$5: [\iarr{10},\iarr{8},\iarr{11},\iarr{13}]$\\
$6: []$\\
$7: [\iarr{12}]$\\
$8: []$
\end{tabular}
};
\node at (7.75,0) {
\begin{tabular}{l}
$E^{\mathrm{arr}}_{\mathrm{dep}}$\\
$\iarr{1}: \idep{1}$\\
$\iarr{2}: \idep{1}$\\
$\iarr{3}: \idep{3}$\\
$\iarr{4}: \idep{1}$\\
$\iarr{5}: \idep{2}$\\
$\iarr{6}: \idep{1}$\\
$\iarr{7}: \idep{2}$\\
$\iarr{8}: \idep{2}$\\
$\iarr{9}: \idep{3}$\\
$\iarr{10}: \idep{1}$\\
$\iarr{11}: \idep{3}$\\
$\iarr{12}: \idep{1}$\\
$\iarr{13}: \idep{4}$\\
\end{tabular}
};
\end{tikzpicture}    
    \end{adjustbox}
    \caption{An example of a temporal graph, where $n=8$, $M=13$, $T=11$. The lower part of the label of each temporal edge indicates its availability time $\tau$ and its traversal time $\lambda$ (hence, the arrival time of the temporal edge is $\tau+\lambda$). The (overlined) upper part of the label of each temporal edge indicates its position in the $E^{\mathrm{arr}}$ list. The $E^{\mathrm{dep}}$, $E_{\mathrm{node}}^{\mathrm{dep}}$, $E^{\mathrm{arr}}_{\mathrm{dep}}$ lists also refer to the (overlined) indexes of $E^{\mathrm{arr}}$. The underlined indices, instead, indicate the position of the corresponding edge in $E^{\mathrm{arr}}$ into the the list $E_{\mathrm{node}}^{\mathrm{dep}}$ of its tail.}
    \label{fig:patg1}
\end{figure*}
    
\subsection{An example of temporal graph}
\label{sec:examplecentrality}
Let us consider the temporal graph in the left part of Figure~\ref{fig:patg1}. The first list in the figure shows $E^{\arr}$ (that is, the list of temporal edges sorted by non-decreasing arrival time), while the second list in the figure shows $E^{\dep}$ (that is, the list of temporal edges sorted by non-decreasing departure time), which is specified by identifying the temporal edges by their position in the list $E^{\arr}$ (the third and four lists are used by our algorithms and are explained in the main text). By assuming $\beta=1$ and by identifying the temporal edges by their (overlined) position in the list $E^{\arr}$, we have that $\langle\iarr{1},\iarr{2},\iarr{4},\iarr{6},\iarr{9}\rangle$ is a $v_{1}v_{5}$-walk which can be extended by either the temporal edge $\iarr{11}$ or the temporal edge $\iarr{13}$. On the contrary, $\langle\iarr{1},\iarr{5}\rangle$ is not a $v_{1}v_{5}$-walk since $\tau_{2} = 4 > 3 = \arr(e_{1}) + \beta$. We also have that $\langle\iarr{7}\rangle$ is the only shortest $v_{1}v_{5}$-walk (with $1$ temporal edge) and the only (shortest) foremost $v_{1}v_{5}$-walk (with arrival time equal to $7$), and $\langle\iarr{3},\iarr{9}\rangle$ is the only (shortest) latest $v_{1}v_{5}$-walk (with departure time equal to $4$ and $2$ temporal edges). Finally, both $\langle\iarr{7}\rangle$ and $\langle\iarr{3},\iarr{9}\rangle$ are fastest $v_{1}v_{5}$-walks (with duration equal to $5$), and only $\langle\iarr{7}\rangle$ is also a shortest fastest $v_{1}v_{5}$-walk. By considering foremost walks, we have that $\sigma^{*}_{v_{1},\iarr{9},v_{6}}=2$, since the the temporal edge $(v_{2},v_{5},6,3)$ (whose index in $E^{\arr}$ is $\iarr{9}$) is contained in the two foremost $v_{1}v_{6}$-walks $\langle\iarr{3},\iarr{9},\iarr{11}\rangle$ and $\langle\iarr{1},\iarr{2},\iarr{4},\iarr{6},\iarr{9},\iarr{11}\rangle$. By setting $\beta=0$ and by considering latest walks, instead, we have that $\sigma^{*}_{v_{1},v_{2}}=\sigma^{*}_{v_{1},\iarr{1},v_{2}}+\sigma^{*}_{v_{1},\iarr{3},v_{2}}+\sigma^{*}_{v_{1},\iarr{6},v_{2}}=1+0+1=2$: as we can see, the latest $v_{1}v_{2}$-walk $\langle\iarr{1},\iarr{2},\iarr{4},\iarr{6}\rangle$ is counted with multiplicity $2$ since it passes twice through $v_{2}$, once via the temporal edge $(v_{1},v_{2},1,1)$ (whose index in $E^{\arr}$ is $\iarr{1}$) and once via the temporal edge $(v_{4},v_{2},5,1)$ (whose index in $E^{\arr}$ is $\iarr{6}$). 

\begin{table*}[ht]
\centering
\begin{adjustbox}{width=\textwidth}
    \begin{tabular}{r|c|c|c|c|c|c|c|c}
 & $v_{1}$ & $v_{2}$ & $v_{3}$ & $v_{4}$ & $v_{5}$ & $v_{6}$ & $v_{7}$ & $v_{8}$  \\
\hline
$v_{1}$ &  & $\{\langle\iarr{1}\rangle,\langle\iarr{3}\rangle\}$ & $\{\langle\iarr{1},\iarr{2}\rangle\}$ & $\{\langle\iarr{1},\iarr{2},\iarr{4}\rangle\}$ & $\{\langle\iarr{7}\rangle\}$ & $\{\langle\iarr{3},\iarr{9},\iarr{11}\rangle\}$ & $\{\langle\iarr{7},\iarr{8}\rangle\}$ & $\{\langle\iarr{7},\iarr{8},\iarr{12}\rangle,\langle\iarr{3},\iarr{9},\iarr{13}\rangle\}$\\
$v_{2}$ & $\emptyset$ &  & $\{\langle\iarr{2}\rangle\}$ & $\{\langle\iarr{2},\iarr{4}\rangle\}$ & $\{\langle\iarr{5}\rangle,\langle\iarr{9}\rangle\}$ & $\{\langle\iarr{5},\iarr{10}\rangle,\langle\iarr{9},\iarr{11}\rangle\}$ & $\{\langle\iarr{5},\iarr{8}\rangle\}$ & $\{\langle\iarr{9},\iarr{13}\rangle\}$  \\
$v_{3}$ & $\emptyset$ & $\{\langle\iarr{4},\iarr{6}\rangle\}$ &  & $\{\langle\iarr{4}\rangle\}$ & $\{\langle\iarr{4},\iarr{6},\iarr{9}\rangle\}$ & $\{\langle\iarr{4},\iarr{6},\iarr{9},\iarr{11}\rangle\}$ & $\emptyset$ & $\{\langle\iarr{4},\iarr{6},\iarr{9},\iarr{13}\rangle\}$ \\
$v_{4}$ & $\emptyset$ & $\{\langle\iarr{6}\rangle\}$ & $\emptyset$ &  & $\{\langle\iarr{6},\iarr{9}\rangle\}$ & $\{\langle\iarr{6},\iarr{9},\iarr{11}\rangle\}$ & $\emptyset$ & $\{\langle\iarr{6},\iarr{9},\iarr{13}\rangle\}$ \\
$v_{5}$ & $\emptyset$ & $\emptyset$ & $\emptyset$ & $\emptyset$ &  & $\{\langle\iarr{10}\rangle,\langle\iarr{11}\rangle\}$ & $\{\langle\iarr{8}\rangle\}$ & $\{\langle\iarr{13}\rangle\}$ \\
$v_{6}$ & $\emptyset$ & $\emptyset$ & $\emptyset$ & $\emptyset$ & $\emptyset$ &  & $\emptyset$ & $\emptyset$ \\
$v_{7}$ & $\emptyset$ & $\emptyset$ & $\emptyset$ & $\emptyset$ & $\emptyset$ & $\emptyset$ &  & $\{\langle\iarr{12}\rangle\}$ \\
$v_{8}$ & $\emptyset$ & $\emptyset$ & $\emptyset$ & $\emptyset$ & $\emptyset$ & $\emptyset$ & $\emptyset$ &  \\
\end{tabular}

\end{adjustbox}
\caption{The set of shortest $st$-walks in the graph of Figure~\ref{fig:patg1}, for any pair of nodes $s$ and $t$}
\label{tbl:setofwalks}
\end{table*}

The shortest walk from $v_1$ to $v_6$ varies in length depending on the value of $\beta$: with $\beta=0$ it is $6$ (which is the length of $\langle\iarr{1},\iarr{2},\iarr{4},\iarr{6},\iarr{9},\iarr{11}\rangle$), with $\beta=1$ it is $3$ (which is the length of $\langle\iarr{3},\iarr{9},\iarr{11}\rangle$), and with $\beta\geq2$ it is $2$ (which is the length of $\langle\iarr{7},\iarr{11}\rangle$). By setting $\beta=2$ and by considering shortest walks, let us compute the $v_{1}$-temporal betweenness of the temporal edge $(v_{4},v_{2},5,1)$ (whose index in $E^{\arr}$ is $\iarr{6}$). Since $\sigma^{*}_{v_{1},\iarr{6},v}=0$ for any $v\neq v_{1}$, we have that $b_{v_{1},\iarr{6}}=0$. However, if we set $\beta=0$, then we have that $\sigma^{*}_{v_{1},\iarr{6},v_{6}}=1$ (since the temporal edge is contained in the unique shortest $v_{1}v_{6}$-walk $\langle\iarr{1},\iarr{2},\iarr{4},\iarr{6},\iarr{9},\iarr{11}\rangle$), while $\sigma^{*}_{v_{1},\iarr{6},v}=0$ for any $v\not\in\{v_{1},v_{6}\}$: hence, in this case, $b_{v_{1},\iarr{6}}=1$. If $\beta=1$, we have that the temporal betweenness vectors are pairwise distinct depending on the chosen measure, as shown in Table~\ref{tab:centralityvalues}.

\begin{table}[ht]
    \centering
    \begin{tabular}{l|c|c|c|c|c|c|c|c}
 & $v_{1}$ & $v_{2}$ & $v_{3}$ & $v_{4}$ & $v_{5}$ & $v_{6}$ & $v_{7}$ & $v_{8}$  \\
\hline
Shortest & 0.0 & 9.5 & 2.0 & 4.0 & 10.0 & 0.0 & 0.5 & 0.0 \\
Foremost & 0.0 & 9.5 & 2.5 & 4.5 & 10.0 & 0.0 & 3.0 & 0.0 \\
Latest & 0.0 & 13.5 & 6.0 & 8.0 & 10.0 & 0.0 & 0.0 & 0.0 \\
Fastest & 0.0 & 10.5 & 2.0 & 4.0 & 10.0 & 0.0 & 0.0 & 0.0 \\
Shortest foremost & 0.0 & 9.0 & 2.0 & 4.0 & 10.0 & 0.0 & 3.0 & 0.0 \\
Shortest latest & 0.0 & 12.0 & 3.0 & 5.0 & 10.0 & 0.0 & 0.0 & 0.0 \\
Shortest fastest & 0.0 & 10.0 & 2.0 & 4.0 & 10.0 & 0.0 & 0.0 & 0.0
\end{tabular}

    \caption{The temporal betweenness vectors with different types of optimal walks for the temporal graph of Figure~\ref{fig:patg1}.}
    \label{tab:centralityvalues}
    \end{table}

Let us, for example, compute the temporal betweenness of node $v_{2}$ in the case of shortest walks. Table~\ref{tbl:setofwalks} shows, for any pair of nodes $s$ and $t$, the set of shortest $st$-walks. Node $v_{2}$ is contained as an inner node in the following walks: the $v_{1}v_{3}$-walk $\langle\iarr{1},\iarr{2}\rangle$, the $v_{1}v_{4}$-walk $\langle\iarr{1},\iarr{2},\iarr{4}\rangle$, the $v_{1}v_{6}$-walk $\langle\iarr{3},\iarr{9},\iarr{11}\rangle$, the $v_{1}v_{8}$-walk $\langle\iarr{3},\iarr{9},\iarr{13}\rangle$, the $v_{3}v_{5}$-walk $\langle\iarr{4},\iarr{6},\iarr{9}\rangle$, the $v_{3}v_{6}$-walk $\langle\iarr{4},\iarr{6},\iarr{9},\iarr{11}\rangle$, the $v_{3}v_{8}$-walk $\langle\iarr{4},\iarr{6},\iarr{9},\iarr{13}\rangle$, the $v_{4}v_{5}$-walk $\langle\iarr{6},\iarr{9}\rangle$, the $v_{4}v_{6}$-walk $\langle\iarr{6},\iarr{9},\iarr{11}\rangle$, and the $v_{4}v_{8}$-walk $\langle\iarr{6},\iarr{9},\iarr{13}\rangle$. All these $10$ walks are the only shortest paths from their corresponding sources to their corresponding destinations, apart from the $v_{1}v_{8}$-walk $\langle\iarr{3},\iarr{9},\iarr{13}\rangle$ (there is also the $v_{1}v_{8}$-walk $\langle\iarr{7},\iarr{8},\iarr{12}\rangle$). Hence, each of them contribute $1$ to the temporal betweenness of $v_{2}$ apart from the $v_{1}v_{8}$-walk which contributes $0.5$. In conclusion, the temporal betweenness of $v_{2}$ is $9.5$ (as shown in the Table~\ref{tab:centralityvalues}).

\section{Computing the \sfo\ betweenness}
\label{sec:algorithms}

In this section we describe an algorithm to compute the $s$-\sfo\ betweenness $b_{s,e}$ of all temporal edges $e$, for a given source node $s$, which runs in time linear in $M$, that is, the number of temporal edges. By repeating the computation for each source and by using Fact~\ref{fact:betweennessformula}, it is then possible to aggregate these $s$-\sfo\ betweennesses to obtain the \sfo\ betweenness of all nodes in time $O(nM)$. The algorithm consists of three phases, a forward, an intermediate, and a backward one. The goal of the forward phase is to count, for each temporal edge $e$, the number of \sh\ $se$-walks and, at the same time, to identify the set of its successors, that is, the set of edges $f$ that can follow $e$ in a \sh\ $st$-walk, for some target node $t$. During the intermediate phase we compute, for each node $v$, the number of \sfo\ $sv$-walks and their cost (that is, the pair including the arrival time and the number of edges in the walk). Finally, the goal of the backward phase is to report the $s$-\sfo\ betweenness of the edges following the successor dependencies.

In the following, given a temporal graph $G=(V,E,\beta)$, we assume that $G$ is represented through its $E^{\mathrm{arr}}$ and $E^{\mathrm{dep}}$ lists (recall that the temporal edges in $E^{\mathrm{dep}}$ are identified by their position in the list $E^{\arr}$). We also assume that, by using these two lists, the following two other lists have been pre-computed, where, once again, each temporal edge is identified by its position in $E^{\mathrm{arr}}$ (see also the examples in Appendices~\ref{sec:examplecentrality} and~\ref{sec:exampleexecution}): $E^{\mathrm{dep}}_\mathrm{node}$, which, for every $v\in V$, contains the list $E^{\mathrm{dep}}_\mathrm{node}[v]$ of temporal edges whose tail is $v$, sorted in non-decreasing order with respect to their departure time, and $E^{\mathrm{arr}}_{\mathrm{dep}}$, which, for every $e=(u,v,\tau,\lambda)\in E$, contains the position of $e$ in $E^{\mathrm{dep}}_\mathrm{node}[u]$ (more precisely, if $e$ is the $i$-th temporal edge in $E^{\mathrm{arr}}$ and the $j$-th temporal edge in $E^{\mathrm{dep}}_\mathrm{node}[u]$, then $E^{\mathrm{arr}}_{\mathrm{dep}}[i]=j$). Note that both $E^{\mathrm{dep}}_\mathrm{node}$ and $E^{\mathrm{arr}}_{\mathrm{dep}}$ can easily be computed in linear-time starting from $E^{\mathrm{arr}}$ and $E^{\mathrm{dep}}$.

\subsection{The non-restless case}
\label{sec:nonrestlessshortestforemost}

We first introduce the algorithm for computing the \sfo\ betweenness in the case $\beta=\infty$, that is, \textit{without waiting constraints} or \textit{non-restless} (see Algorithm~\ref{alg:sfononrestless}), whose forward phase is built upon the algorithm of~\cite{BrunelliV2022} for computing single-source minimum-cost walks.

\begin{algorithm}[ht]
\small
\Input{$G=(V,E,\infty)$ (represented by $E^{\mathrm{dep}}$ and $E^{\mathrm{arr}}$) and $s\in V$}
\Output{$s$-\sfo\ betweenness $b_{s,e}$, for all $e\in E$}

Compute the lists $E^{\mathrm{dep}}_{\mathrm{node}}$ and $E^{\mathrm{arr}}_{\mathrm{dep}};$\\
\lForEach{$v\in V$}{$l[v]:=1;c[v]:=\infty;\sigma[v]:=0;$} 
\lForEach{$e\in E$}{$L[e]:=0;C[e]:=\infty;\Sigma[e]:=0;$}
\ForEach{$e = (u,v,\tau,\lambda)\in E^{\mathrm{arr}}$}{\label{alg2:forwardstart}
    \lIf{$E^{\mathrm{arr}}_{\mathrm{dep}}[e]\geq l[u]$}{$\processcosts(u,E^{\mathrm{arr}}_{\mathrm{dep}}[e]);$}
    \lIf{$u=s$}{$C[e]:=1; \Sigma[e]:=1;$}
    \If{$C[e]\neq\infty$}{
        \If{$C[e]\leq c[v]$}{
            $a=l[v];D := E^{\mathrm{dep}}_{\mathrm{node}}[v];$\\
            \lWhile{$a\leq|D|\wedge\dep(E^{\mathrm{arr}}[D[a]]) < \tau+\lambda$}{$a:=a+1;$}
            $\processcosts(v,a-1);$\\
            \lIf{$C[e]<c[v]$}{$c[v]:=C[e]; \sigma[v]:=0;$}
            $\sigma[v] := \sigma[v] + \Sigma[e]$; $L[e]:=a;$\\
        }
    }
\label{alg2:forwardend}}
\lForEach{$v\in V$}{$c^{*}[v]:=[\infty,\infty];\sigma^{*}[v]:=0;\delta[v]:=0;$} 
\lForEach{$e\in E$}{$\Sigma^{*}[e]:=0;b[e]:=0;$}
\ForEach{$e = (u,v,\tau,\lambda)\in E^{\mathrm{arr}}$\label{alg2:pre-backward1}}{\If{$C[e]<\infty\wedge [\arr(e),C[e]]\vartriangleleft c^{*}[v]$}{$c^{*}[v]:=[\arr(e),C[e]];$\label{alg2:pre-backward2}}}
\ForEach{$e = (u,v,\tau,\lambda)\in E^{\mathrm{arr}}$\label{alg2:pre-backward3}}{\If{ $[\arr(e),C[e]]=c^{*}[v]$}{$\Sigma^{*}[e]:=\Sigma[e]$; $\sigma^{*}[v]:=\sigma^{*}[v]+\Sigma[e];$\label{alg2:pre-backward4}}}
\ForEach{$e = (u,v,\tau,\lambda)\in \mathrm{reverse}(E^{\mathrm{arr}}):L[e]>0$}{\label{alg2:backwardstart}
    \If{$c^{*}[v]\vartriangleleft [\arr(e),C[e]]$ }{$\delta[v]:=0; c^{*}[v]:=[\arr(e),C[e]];$ }\label{alg2:zerodelta}
    \lFor{$f\in E_{\mathrm{node}}^{\mathrm{dep}}[v][L[e]:l[v]-1]$}{$\delta[v] := \delta[v] + b[f]/\Sigma[f];$}
    $l[v]:=L[e]; b[e]:=\Sigma[e]\delta[v];$\\
    \lIf{$\Sigma^{*}[e]>0$}{$b[e]:=b[e]+\Sigma^{*}[e]/\sigma^{*}[v];$}
\label{alg2:backwardend}}
\Return $b$ 
\BlankLine
\SetArgSty{textbf}
\Finalize{$u,j$}:\\
\Indp
    \If{$c[u]\neq\infty$}{\label{alg2:finalizestart}
        \lForEach{$f\in E^{\mathrm{dep}}_{\mathrm{node}}[l[u],j]$}{$C[f]:=c[u]+1; \Sigma[f]:=\sigma[u];$}
    }
    $l[u]:=j+1;$\label{alg2:finalizeend}
\caption{Compute non-restless \sfo\ $b_{s,e}$, for all $e\in E$}
\label{alg:sfononrestless}
\end{algorithm}

\noindent\textbf{Forward phase (lines~\ref{alg2:forwardstart}-\ref{alg2:forwardend}).} Let $G_{k}=(V,E_{k},\infty)$ be the temporal graph containing only the first $k$ temporal edges in $E^{\mathrm{arr}}$. The algorithm scans the edges in $E^{\mathrm{arr}}$ one after the other and after scanning $k$ edges, for each node $v\in V$, it updates the following three values: the length $c[v]$ of any \sh\ $sv$-walk in $G_{k}$, the number $\sigma[v]$ of these walks in $G_{k}$, and the position $l[v]$ in $E^{\mathrm{dep}}_\mathrm{node}[v]$ such that all temporal edges in $E^{\mathrm{dep}}_\mathrm{node}[v]$ starting from this position can extend a \sh\ $sv$-walk in $G_{k}$. Note that $c[v]$ can only decrease as $k$ increases, while $l[v]$ can only increase.

At the beginning, for each node $v$, $c[v]=\infty$, $\sigma[v]=0$, and $l[v]=1$. Suppose that the first $k-1$ temporal edges have been scanned, and that the edge $e_{k}=(u,v,\tau,\lambda)$ has now to be analysed (the reader can also refer to the figure shown in Appendix~\ref{sec:nrtsfob}). Let us first analyse $e_{k}$ from the point of view of its tail, that is, $u$. Clearly, $e_{k}$ is included in $E^{\mathrm{dep}}_\mathrm{node}[u]$: suppose that it appears in position $i$. If $i\geq l[u]$, then all temporal edges between position $l[u]$ and $i$, which have a departure time not greater than $\mathrm{dep}(e_{k})$, cannot extend any $su$-walk ending with a temporal edge following $e_{k}$ in $E^{\mathrm{arr}}$, since such a temporal edge has arrival time greater than $\mathrm{dep}(e_{k})$. Each such edge $f$ can then be ``finalised'' (lines~\ref{alg2:finalizestart}-\ref{alg2:finalizeend}), that is, the length $C[f]$ of any \sh\ $sf$-walk (in $G_{k}$) can be set equal to $c[u]+1$ and the number $\Sigma[f]$ of these walks (in $G_{k}$) can be set equal to $\sigma[u]$. We use the term ``finalise'' to emphasise that $C[f]$ will be the same in $G_{k'}$ for $k'\ge k$.  We can also set $l[u]=i+1$, since all temporal edges in $E^{\mathrm{dep}}_\mathrm{node}[u]$ starting from position $i+1$ can still extend a \sh\ $su$-walk in $G_{k'}$ for $k'\ge k$. Moreover, if $u=s$, then we can set $C[e_{k}]$ equal to $1$ (since we are considering \sh\ walks from $s$) and the number of \sh\ walks ending with $e_{k}$ is equal to $1$ (since there is only one such walk, that is, $\langle e_{k}\rangle$).

Let us now analyse the temporal edge $e_{k}$ from the point of view of its head, that is, $v$, by assuming that there exists at least one \sh\ walk ending with $e_{k}$ in $G_{k}$. If $C[e_{k}]$ is not greater than $c[v]$ (that is, $e_{k}$ ends a \sh\ $sv$-walk), we first compute the first position $a$ in $E^{\mathrm{dep}}_\mathrm{node}[v]$ of a temporal edge whose departure time is at least equal to $\tau+\lambda$ (that is, a temporal edge which can extend an $se$-walk in $G_{k}$). All the temporal edges in $E^{\mathrm{dep}}_\mathrm{node}[v]$ between the position $l[v]$ and the position $a-1$ can now be finalised, since they cannot be the successor of any temporal edge of $E^{\mathrm{arr}}$ following $e_{k}$. We then set $l[v]=a$, since all temporal edges in $E^{\mathrm{dep}}_\mathrm{node}[v]$ starting from position $a$ can still extend a \sh\ $sv$-walk in $G_{k'}$ for $k'\ge k$ (in particular, they extend \sh\ $se$-walks). Moreover, if adding the temporal edge $e_{k}$ to $G_{k-1}$ reduces the length $c[v]$ of the \sh\ $sv$-walks (that is, $C[e_{k}] < c[v]$), then we have to update $c[v]$, by setting it equal to $C[e_{k}]$ (that is, the length of the \sh\ $se_{k}$-walks), and $\sigma[v]$ by setting it equal to the number $\Sigma[e_{k}]$ of \sh\ $se_{k}$-walk in $G_{k}$. Otherwise (that is, $C[e_{k}] = c[v]$ and adding the temporal edge $e_{k}$ to $G_{k-1}$ does not change the length $c[v]$ of the \sh\ $sv$-walks), the number $\Sigma[e_{k}]$ of \sh\ $se_{k}$-walk in $G_{k}$ has to be added to $\sigma[v]$ (since all \sh\ $se_{k}$-walks are also \sh\ $sv$-walks). Finally, we store in $L[e_{k}]$ the position $a$ in $E^{\mathrm{dep}}_\mathrm{node}[v]$, which is the first position of the successors of $e$ (to be used in the backward phase).

\noindent\textbf{Intermediate phase (lines~\ref{alg2:pre-backward1}-\ref{alg2:pre-backward4}).} Once we have computed, for each $e\in E$, the length $C[e]$ of any \sh\ $se$-walk and the number $\Sigma[e]$ of these walks, it is easy to compute, for each $v\in V$, $c^*[v]$, where $c^*[v]$ specifies both the arrival time $c^*[v][1]$ and the length $c^*[v][2]$ of any \sfo\ $sv$-walk. Indeed, it suffices to scan the temporal edges in $E^{\mathrm{arr}}$ and, for each edge $e=(u,v,\tau,\lambda)$, to verify whether $(\arr(e),C[e])\vartriangleleft c^*[v]$ in which case $c^*[v]$ has to be set equal to $(\arr(e),C[e])$ (lines~\ref{alg2:pre-backward1}-\ref{alg2:pre-backward2}). Once $c^*[v]$ has been computed for each $v\in V$, the number $\sigma^*[v]$ of \sfo\ $sv$-walks can also be computed. Indeed, it suffices to scan again the temporal edges in $E^{\mathrm{arr}}$ and, for each edge $e=(u,v,\tau,\lambda)$, to verify whether $(\arr(e),C[e])=c^*[v]$ in which case $\sigma^*[v]$ has to be increased by the value  $\Sigma[e]$ (lines~\ref{alg2:pre-backward3}-\ref{alg2:pre-backward4}). Note that, at the same time, we can also compute the number $\Sigma^*[e]$ of \sfo\ $se$-walks.

\noindent\textbf{Backward phase (lines~\ref{alg2:backwardstart}-\ref{alg2:backwardend}).} The backward phase simply applies Lemma~\ref{lem:backwardlemma} in a ``reverse'' way, by scanning the temporal edges in $E^{\mathrm{arr}}$ from the last to the first one and, for each scanned edge $e$, by accumulating on its head $v$ the contribution to $b_{s,e}$ of each successor of $e$. More precisely, we store in $\delta[v]$ the partial sum $\sum_{f\in E_{\mathrm{node}}^{\mathrm{dep}}[v][l[v]:z]}\frac{b[f]}{\Sigma[f]}$ where $z$ denotes the last index when the sum was zeroed. It can be updated in constant time per edge $f\in E_{\mathrm{node}}^{\mathrm{dep}}[v]$ each time we encounter an edge $e$ with head $v$. Note that each such edge $e$ has successors $E_{\mathrm{node}}^{\mathrm{dep}}[v][L[e]:r]$, where $r$ is the position of the last edge $f$ in $E_{\mathrm{node}}^{\mathrm{dep}}$ such that $C[f]=C[e]+1$ and the index $L[e]$ of the first successor can only decrease as we scan edges $e$ with lower arrival times.\footnote{As it is common in several programming languages, given a sequence $A$ and two positive integers $l$ and $r$ both not greater than the length of $A$ and such that $l\leq r$, we denote by $A[l:r]$ the sub-sequence of $A$ from position $l$ to position $r$, both included. Moreover, $A[l:]=A[l:|A|]$.} As we scan the $k$th edge $e$ in $E^{\mathrm{arr}}$, we also maintain in $c^*[v]$ the arrival time and the length of any \sfo\ $sv$-walk in $G_k$.  Whenever $c^*[v]\vartriangleleft(\arr(e),C[e])$, the index $r$ for $e$ is indeed $l[v]-1$ and the quantity accumulated on $v$ has to be zeroed while $c^*[v]$ has to be updated to $(\arr(e),C[e])$.

\begin{theorem}\label{thm:sfo_alg1analysis}
For any temporal graph $G=(V,E,\infty)$ and for any $s\in V$, Algorithm~\ref{alg:sfononrestless} correctly computes the $s$-\sfo\ betweenness $b_{s,e}$ in time $O(M)$.
\end{theorem}

In Appendix~\ref{sec:proofs2} we prove the above theorem, while in Appendix~\ref{sec:exampleexecution} we describe an example of execution of Algorithm~\ref{alg:sfononrestless} on a simple temporal graph, and in Appendix~\ref{sec:shortestnonrestless} we show how this algorithm can be easily simplified in order to compute the non-restless \sh\ betweenness.

\subsection{The restless case}
\label{sec:general}

We now briefly describe how Algorithm~\ref{alg:sfononrestless} has to be modified in order to deal with the general case (that is, $\beta\leq\infty$): the new algorithm manages the increased complexity of the restless constraint through appropriate lists of interval quintuples which correspond to windows of time with temporal edges from a node that extend the same optimal walks. Once again, we execute a forward phase (in order to compute, for each temporal edge $e$, the length $C[e]$ of any \sh\ $se$-walk and the number $\Sigma[e]$ of these walks), followed by the same intermediate phase and a backward phase in which Lemma~\ref{lem:backwardlemma} is applied. Note that in Algorithm~\ref{alg:sfononrestless}, during the forward phase, in order to correctly apply the lemma in the backward phase, it sufficed to memorize, for each temporal edge $e$, the position $L[e]$ in the list $E^{\mathrm{dep}}_\mathrm{node}[v]$ of the first temporal edge which could extend a \sh\ $se$-walk. This was due to the fact that, in the non-restless case, if an en edge $f$ in $E^{\mathrm{dep}}_\mathrm{node}[v]$ extends a $se$-walk, then all edges following $f$ in $E^{\mathrm{dep}}_\mathrm{node}[v]$ also extend the walk. In the general case, this is not true anymore (because of the waiting constraints) and the forward phase of the algorithm has to maintain additional information to be used during the backward phase. In particular, given a temporal graph $G=(V,E,\beta)$, for each node $v\in V$, the general algorithm maintains a list $\intervs{v}$ of \textit{interval quintuples} $Q=(l,r,c,P,\eta)$ where $1\leq l,r\leq \left|E^{\mathrm{dep}}_\mathrm{node}[v]\right|$, $c\in\mathbf{N}$, $\eta\in\mathbf{N}$, and $P\subseteq E$ is a list of predecessor edges. The semantic of an interval quintuple $Q$ is the following.

\begin{itemize}
    \item $l$ and $r$ are the left and right extremes of an interval $Q.I$ of edges in $E^{\mathrm{dep}}_\mathrm{node}[v]$ (that is, $Q.I=E^{\mathrm{dep}}_\mathrm{node}[v][l:r]$).
    
    \item $c$ is the length of any \sh\ walk from the source $s$ (in the temporal graph induced by the edges scanned so far), such that edges in $Q.I$ extend it.
    
    \item The edges ending these walks are predecessors of edges in $Q.I$ and are stored in the ordered list $P$ (sorted by arrival times). More precisely, for any edge $f=Q.I[i]$, with $1\leq i\leq r-l+1$, the set of edges $P_{i}$, that precede $f$ in the above mentioned \sh\ walks, is a subset of $P$. We rely on the fact that edges with the same predecessors partition $Q.I$ into consecutive intervals. Indeed, $P_{i+1}$ is included in $P_{i}$ if $i\le r-l$. The ordering of $P$ respect the inclusion ordering $P_1\supseteq P_2\supseteq\cdots$ so that each $P_i$ is a suffix of $P$.
    
    \item $\eta$ is the total number of \sh\ walks that end with an edge in $P$ (they all have length $c$).
\end{itemize}
The forward and backward phase of the general algorithm as well as the finalisation of an edge are more complicated than in the case of the non-restless case, in order to deal with the list of interval quintuples. A detailed description of these modifications along with the pseudo-code of the algorithm and the correctness and complexity analysis are presented in Appendix~\ref{sec:sforestless}.

\section{Computing other betweennesses}

In Appendix~\ref{sec:generalalg} we introduce a general cost framework and we show how it can  encompass the following optimality criteria: fastest (\fa), foremost (\fo), \sh, and shortest fastest (\sfa). Other criteria (such as latest) can also be dealt with by the framework (we leave to the reader the freedom of choosing other such criteria). In the same appendix we then prove that the algorithm for the restless \sfo\ case can be extended to any criteria that fits in the framework.

\section{The experiments}

In this section we perform several experimental analysis in order to compare the performance of our algorithms with respect to the ones previously proposed in the literature, and in order to apply the algorithms themselves to a specific case study in the field of public transport networks. Our experimental study includes the following three algorithms.

\begin{itemize}
\item \bmnr: this is the algorithm proposed in~\cite{BussMNR20} to compute the exact values of the \sh\ and \sfo\ betweenness of all nodes.

\item \onbra: this is the approximation algorithm proposed in~\cite{Santoro_2022}, which is based on a sampling technique for obtaining an absolute approximation of the \sh\ betweenness values (both in the non restless and in the restless case).

\item \bcl: this is the algorithm described in the previous sections, that we use here to compute the exact values of the \sh\ (Algorithm~\ref{alg:shnonrestless}) and \sfo\ (Algorithm~\ref{alg:sfononrestless}) betweenness of all nodes.
\end{itemize}

In order to analyse the correlation between different centrality rankings, we used three different metrics: the Kendall's $\tau$ correlation coefficient~\cite{Kendall38}, a weighted version of this coefficient~\cite{Vigna15}, and the intersection of the top-$1000$ ranked nodes (note that the latter metric is directly translatable into the Jaccard similarity of the top ranked nodes). For the weighted Kendall's $\tau$ coefficient, we used the hyperbolic weighting scheme, that gives weights to the positions in the ranking which decay harmonically with the ranks, i.e., the weight of rank $r$ is $1/(r+1)$. Both the Kendall's $\tau$ correlation coefficient and its weighted version has been computed by using the Java code available at the Laboratory for Web Algorithmics~\cite{LAWWebSite}.

\subsection{Comparing algorithms execution times}
\label{sec:experiment1}

First, we compare the running time of \bmnr, of \onbra, and of \bcl. We present here the results concerning only the computation of the \sh\ betweenness: however, the results are similar in the case of the \sfo\ betweenness (see the table in Section~\ref{sec:experiment1sfb} in the appendix). We used the Julia implementation of \bmnr\ associated to~\cite{Becker2023} and available at~\cite{TSBProxyWebSite}.\footnote{As stated in~\cite{Becker2023}, the original C implementations of \bmnr\ caused overflow (indicated by negative centralities) and out of memory errors.} We also implemented \bcl\ in Julia: our code is available at \url{https://github.com/piluc/TWBC/}. A Rust implementation is also available at \url{https://github.com/lviennot/tempograph/}. Finally, we made use of the results reported in \cite{Santoro_2022} which, in turn, made use of the C implementation of \onbra\ available at~\cite{ONBRAWebSite}. We executed the experiments on a server running Ubuntu 20.04.5 LTS 112 with processors Intel(R) Xeon(R) Gold 6238R CPU @ 2.20GHz and 112GB RAM.

\noindent\textbf{Dataset.} In this first experiment, we execute the algorithms on the set of temporal graphs used in~\cite{Becker2023}, which includes almost all the networks of~\cite{BussMNR20} and of~\cite{Santoro_2022}. As stated in~\cite{Becker2023}, this set does not include one temporal graph used in~\cite{BussMNR20}, because it does not appear to be available anymore, and it replaces one temporal graph used in~\cite{Santoro_2022} by a bigger one from a different domain to make the set of analyzed temporal graphs more diverse. This latter network excluded in~\cite{Becker2023} and the two larger temporal graphs analyzed in~\cite{Santoro_2022}, which were excluded in~\cite{Becker2023} because of the excessive amount of time needed to compute their exact temporal betweenness values, will be analysed in our second experiment. The properties of these networks are summarized in the first five columns of Table~\ref{tab:experiment1}.

\begin{table}[t]
\centering
\begin{tabular}{@{}lrrrcrrr@{}}\toprule
\textbf{Network}    & \multicolumn{1}{c}{$\mathbf{n}$} & \multicolumn{1}{c}{$\mathbf{M}$} & \multicolumn{1}{c}{$\mathbf{T}$} & \textbf{URL} & \multicolumn{1}{c}{$\mathbf{t}_{\mathbf{\textsc{BMNR}}}$} & \multicolumn{1}{c}{$\mathbf{t}_{\mathbf{\textsc{Fast}}}$} & \multicolumn{1}{c}{$\frac{\mathbf{t}_{\mathbf{\textsc{BMNR}}}}{\mathbf{t}_{\mathbf{\textsc{Fast}}}}$} \\\midrule
\texttt{Infectious} & 10972 & 831824 & 76944 & \cite{SocioPatternsWebSite} & 3111.19 & 1603.53 & 1.94 \\
\texttt{DiggReply} & 30360 & 86203 & 82641 & \cite{TemporalNetworkRepositoryWebsite} & 1190.41 & 506.53 & 2.35 \\
\texttt{FacebookWall} & 35817 & 198028 & 194904 & \cite{TemporalNetworkRepositoryWebsite} & 3410.98 & 1317.76 & 2.59 \\
\texttt{SMS} & 44090 & 544607 & 467838 & \cite{TemporalNetworkRepositoryWebsite} & 12476.04 & 4721.05 & 2.64 \\
\texttt{SlashdotReply} & 51083 & 139789 & 89862 & \cite{TemporalNetworkRepositoryWebsite} & 4489.14 & 1643.77 & 2.73 \\
\texttt{WikiElections} & 7115 & 106985 & 101012 & \cite{SnapnetsWebSite} & 521.11 & 113.68 & 4.58 \\
\texttt{CollegeMsg} & 1899 & 59798 & 58911 & \cite{SnapnetsWebSite} & 242.81 & 23.20 & 10.47 \\
\texttt{Topology} & 16564 & 198038 & 32823 & \cite{KonectWebSite} & 8704.15 & 792.82 & 10.98 \\
\texttt{Hypertext09} & 113  & 41636 & 5246 & \cite{SocioPatternsWebSite} & 83.63 & 0.80 & 104.37 \\
\texttt{HighSchool11} & 126 & 57078 & 5609 & \cite{SocioPatternsWebSite} & 132.33 & 1.18 & 111.86 \\
\texttt{HighSchool12} & 180 & 90094 & 11273 & \cite{SocioPatternsWebSite} & 394.70 & 2.93 & 134.89 \\
\texttt{PrimarySchool} & 242 & 251546 & 3100 & \cite{SocioPatternsWebSite} & 1895.45 & 12.92 & 146.67 \\
\texttt{EmailEU} & 986 & 327336 & 207880 & \cite{SnapnetsWebSite} & 11942.84 & 72.24 & 165.32 \\
\texttt{HighSchool13} & 327 & 377016 & 7375 & \cite{SocioPatternsWebSite} & 7131.86 & 29.63 & 240.67 \\
\texttt{HospitalWard} & 75 & 64848 & 9453 & \cite{SocioPatternsWebSite} & 204.43 & 0.83 & 247.49 \\
\bottomrule
\end{tabular}
\caption{The temporal graphs used in our first experiment, where $n$ denotes the number of nodes, $M$ the number of temporal edges, $T$ the number of unique time steps, $\mathrm{t_{\textsc{BMNR}}}$ the execution time of \textsc{BMNR}, and $\mathrm{t_{\textsc{Fast}}}$ the execution time of \textsc{Fast}.}
\label{tab:experiment1}
\end{table}

\noindent\textbf{Results.} The execution times of \bmnr\ and of \bcl\ are shown in the 6th and 7th columns of Table~\ref{tab:experiment1}. As it can be seen, \bcl\ is between approximately two and almost 250 times faster than \bmnr\ (see the eighth column). It is worth observing that the execution time of our Julia implementation of \bcl\ is significantly lower than the execution time (reported in~\cite{Santoro_2022}) of the C implementation of \onbra, which was executed on an architecture not very different from ours. Indeed, on the four temporal graphs of our dataset that have also been used in the experimental analysis of~\cite{Santoro_2022} (that is, \texttt{College msg}, \texttt{Email EU}, \texttt{Facebook wall}, and \texttt{SMS}), the execution time of \onbra\ is approximately 6, 25, 2, and 3 times slower than our algorithm (note that \onbra\ computed the estimates of the \sh\ betweenness values, by using a sample of node pairs whose size was less than $1\%$ of the number of all node pairs).

\begin{table*}
\centering
\begin{adjustbox}{width=\textwidth}
\begin{tabular}{@{}lrrrcrrrrrr@{}}\toprule
\textbf{Temporal graph}    & \multicolumn{1}{c}{$\mathbf{n}$} & \multicolumn{1}{c}{$\mathbf{M}$} & \multicolumn{1}{c}{$\mathbf{T}$} & \textbf{Source} & \multicolumn{1}{c}{$\mathbf{t}_{\bmnr}$} & \multicolumn{1}{c}{$\mathbf{t}_{\bcl}$} & \multicolumn{1}{c}{$\frac{\mathbf{t}_{\bmnr}}{\mathbf{t}_{\bcl}}$} & \multicolumn{1}{c}{$\mathbf{t}_{\onbra}$} & \multicolumn{1}{c}{Sample size} & \multicolumn{1}{c}{Weighted $\tau$}\\\midrule
\texttt{MathOverflow} & 24759 & 390414 & 389952 & \cite{SnapnetsWebSite} & 46594 & 2117 & 22.01 & 36983 & $30650$ & $0.88$\\
\texttt{AskUbuntu} & 157222 & 726639 & 724715 & \cite{SnapnetsWebSite} & 421781 & 32280 &  13.07 & 35585 & $14831$ & $0.86$\\
\texttt{SuperUser} & 192409 & 1108716 & 1105102 & \cite{SnapnetsWebSite} & 972104 & 63553 & 15.30 & 41856 & $11106$ & $0.86$\\
\bottomrule
\end{tabular}
\end{adjustbox}
\caption{The temporal graphs used in our second experiment, where $t_{\bmnr}$ denotes the execution time of \textsc{BMNR}, $\mathrm{t_{\textsc{Fast}}}$ the execution time of \textsc{Fast}, and $t_{\onbra}$ the execution time of \onbra\ reported in~\cite{Santoro_2022}. The last two columns show the sample size used by \onbra\ and the weighted Kendall $\tau$ coefficient of the ranking produced by \onbra\ with this sample size, respectively.}
\label{tab:experiment2}
\end{table*}

\subsection{Analysing three larger temporal graphs}

In the experimental analysis of~\cite{Santoro_2022}, the authors consider three other temporal graphs, whose properties are summarised in the first five columns of Table~\ref{tab:experiment2}. According to the authors, on these temporal graphs the algorithm \bmnr\ was not able to conclude the computation on their machine since it required too much memory, while \onbra\ could provide estimates of the \sh\ betweenness centrality values in the time indicated in the eighth column of the table. The memory problems of \bmnr\ have been already solved in the Julia implementation of \bmnr, by using dictionaries instead of matrices. Hence, we have been able to execute the algorithm \bmnr\ with input these three temporal graphs. By using \bcl, the computation of the \textit{exact} \sh\ betweenness values requires approximately 0.6, 9, and 17.7 hours, thus significantly improving over \bmnr\ (very similar results hold in the case of the \sfo\ betweenness). Moreover, the execution time of our Julia implementation of \bcl\ on the three networks is significantly less than, comparable with, and approximately 1.5 bigger than the reported execution time of the C implementation of \onbra\ (recall that \onbra\ computes estimates of the betweenness values). The second to last column of Table~\ref{tab:experiment2} shows the sample size used by \onbra, which is the one reported in~\cite{Santoro_2022}, while the last column shows the weighted Kendall $\tau$ coefficient  between the ranking produced by \bcl\ and the one produced by \onbra\ with this sample size. As it can be seen, in order to be competitive in terms of execution time, \onbra\ produces centrality values which are quite imprecise in terms of rankings.

\begin{figure}[ht]
\centering
\begin{adjustbox}{width=0.7\textwidth}
\includegraphics{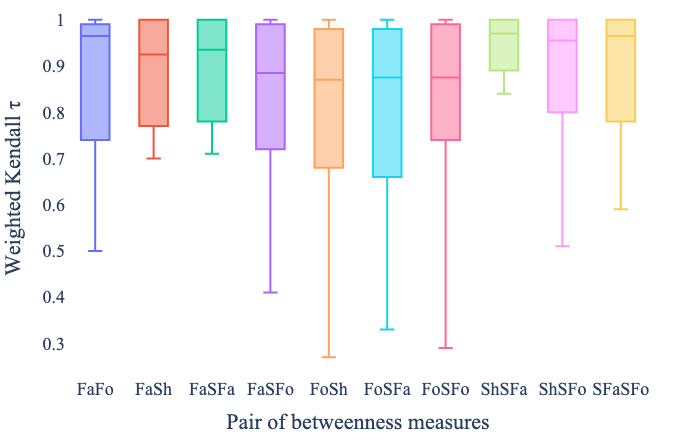}
\end{adjustbox}
\caption{The quartiles of the weighted Kendall $\tau$ over 14 networks, for all pairs of betweenness measures, with $\beta=2400$.}
\label{fig:boxplot3}
\end{figure}

\noindent\textbf{Comparing \sh\ and \sfo\ betweennesses.} Once we have computed the \sh\ and \sfo\ betweenness values of the nodes of the above three larger temporal graphs, we analysed the correlation between these two measures. These correlations are very high, especially if we consider the weighted Kendall's $\tau$ coefficient or the intersection of the top-$1000$ ranked nodes, whose values are $0.97$, $0.98$, and $0.98$, and $949$, $960$, and $962$, respectively. In other words, if we are interested in the top nodes in the rankings, then there is not so much difference between using the \sh\ and the \sfo\ betweenness measure.

\subsection{Analysing ranking correlations}

\begin{figure}[ht]
\centering
\begin{adjustbox}{width=0.7\textwidth}
\includegraphics{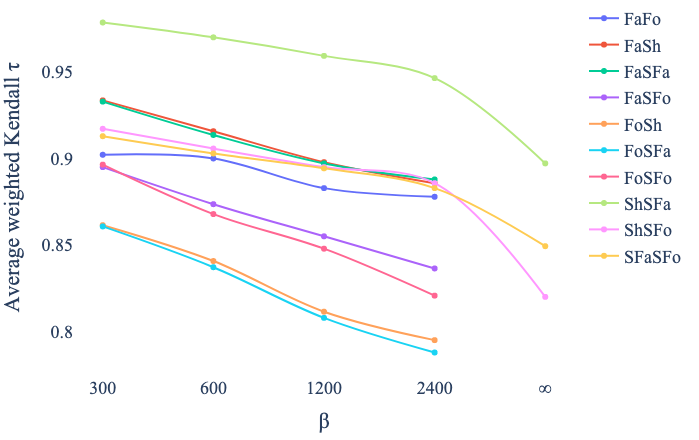}
\end{adjustbox}
\caption{The average weighted Kendall $\tau$ over 14 networks, for all pairs of betweenness measures, as a function of the value of $\beta$.}
\label{fig:avgwktvsbeta}
\end{figure}

Our third experiment consists of comparing the rankings of the nodes of a temporal graph when sorted according to their betweenness values, computed with different waiting constraints and for different walk optimality criteria. In particular, we considered the following values of $\beta$: $300$, $600$, $1200$, $2400$, and $\infty$. Moreover, we computed the \fa, \fo, \sh, \sfa, and \sfo\ betweenness values (by using the general \bcl\ algorithm, that is, Algorithm~\ref{alg:restless} in the appendix). However, we did not compute the \fo\ and the \fa\ betweenness values in the case $\beta=\infty$, due to the huge number of optimal walks: this involves using Julia number data structures which causes our algorithm to take an excessive amount of time needed to calculate the exact betweenness values. For the very same reason (for different values of $\beta$), we excluded the \texttt{Topology} network from the dataset used in this experiment, which otherwise is the same as the one used in the first experiment.

The values of the Kendall $\tau$ correlation, of its weighted variant, and of the top-50 intersection are shown in the tables included in Appendices~\ref{sec:rankingcorrelation}: we here discuss the weighted Kendall $\tau$ results only. From the tables in the appendix, we can conclude that all the betweenness measures are highly correlated, with the \sh\ and the \sfa\ betweenness being the two more correlated (see the box plot of Figure~\ref{fig:boxplot3}, which show the quartiles of the correlation values for $\beta=2400$). In particular, in half of the networks in the dataset the weighted Kendall $\tau$ for this pair of betweenness measures with $\beta=2400$ is at least $0.97$ and the minimum weighted Kendall $\tau$ is $0.84$. This minimum value is reached in correspondence of the \texttt{Primary school} network, for which the weighted Kendall $\tau$ is the minimum one for all the pairs of betweenness measures.

Another interesting observation is that the weighted Kendall $\tau$ values seem to depend on the waiting constraints. In particular, we can observe that these values tend to decrease as the value of $\beta$ increases, as shown in the left part of Figure~\ref{fig:avgwktvsbeta} where the average weighted Kendall $\tau$ values over the fourteen networks is shown as a function of the value of $\beta$. This behaviour might be justified by the fact that when $\beta$ increases the set of optimal walks with respect to different betweenness measures may be quite different: that is, the more stringent are the waiting constraints the more similar are the sets of optimal walks.

\subsection{Analysing public transport networks}

All the networks analysed so far are ``uniform'' temporal graphs, in the sense that the traversal time of all temporal edges is equal to $1$. In this last experiment, instead, we use a subset of the dataset published in~\cite{Kujala2018} and used~\cite{Crescenzi2019}. This dataset includes 25 cities’ public transport networks and is available in multiple formats including the temporal edge list for a specific working day (note that each temporal edge has a travel time usually much greater than $1$). The list of the cities that we used in our experiment is summarized in Table~\ref{tab:cities}, where, for each city, we provide, in the first two columns, the number of nodes (that is, the number of stops) and the number of temporal edges. In this experiment, we focus on the \sfa\ and the \sfo\ betweenness values, because of two main reasons. First, each temporal edge of the temporal graph relates to the connection between two stations of a transport trip: hence, counting the number of temporal edges in a walk does not indicate the number of transfers (which instead should be more interesting to analyse in the case of a public transport network). For this reason, we have not analysed the \sh\ betweenness measure. Secondly, focusing on walks which are the shortest among the fastest and the foremost ones allows us to analyse walks which are the closest to being paths in the case of waiting constraints (which are both desirable properties in the case of public transport trips).

\begin{table}[b]
\centering
\begin{tabular}{lrrrrrrr}
\toprule
\multicolumn{1}{c}{\textbf{City}} & \multicolumn{1}{c}{$\mathbf{n}$} & \multicolumn{1}{c}{$\mathbf{M}$} & \multicolumn{1}{c}{$\mathbf{m}$} & \multicolumn{1}{c}{$\mathbf{M/m}$} & \multicolumn{1}{c}{$\mathbf{t_{\bcl}}$} & \multicolumn{1}{c}{$\mathbf{t_{B}}$}\\
\midrule
Berlin & $4601$ & $1048209$ & $12359$ & $85$ & $35555$ & $8.03$\\
Bordeaux & $3435$ & $236595$& $4040$ & $59$ & $2336$ & $2.95$\\
Kuopio & $549$ & $32117$ & $979$  & $33$ & $126$ & $0.07$\\
Paris & $11950$ & $1823871$ & $16704$ & $109$ & $121157$ & $33.51$\\
Rome & $7869$ & $1051211$ & $10143$ & $104$ & $70236$ & $18.77$ \\
Venice & $1874$ & $118480$ & $3464$ & $34$ & $439$ & $0.67$\\
\bottomrule
\end{tabular}
\caption{The six cities analysed in our case study. The value $m$ denotes the number of edges in the underlying graph, while the value $t_{\mathrm{B}}$ denotes the execution time of the Brandes' algorithm on this graph.}
\label{tab:cities}
\end{table}

Given a temporal graph $G=(V,E,\beta)$, the \textit{underlying graph} of $G$ is the graph whose set of nodes is $V$ and whose set of edges contains all pairs $(u,v)$ such that $(u,v,\tau,\lambda)\in E$, for some $\tau$ and $\lambda$. The main goal of this experiment is to verify how much the \sfa\ and the \sfo\ betweenness measures of a temporal graph are correlated to the (classical) betweenness measure of the corresponding underlying graph. In the second to last column of the table we show the average execution time of the \bcl\ algorithm computing the \sfa\ and the \sfo\ betweenness with $\beta=300,600,1200,2400,\infty$, while the last column shows the execution time for computing the betweenness values of the underlying graph by using the Brandes' algorithm~\cite{Brandes_2001} (as implemented in the Julia \texttt{Graphs} package~\cite{GraphsWebSite}). As it can be seen, this latter algorithm is significantly faster than the \bcl\ algorithm. Actually, we might expect an additional multiplicative factor close to the ratio $M/m$, where $m$ denotes the number of edges in the underlying graph. In practice, this factor is between 20 and 50 times bigger because of two main reasons: on the one hand, our code for the general case is not as optimised as the code for the non-restless case thus leading to a code around ten times slower, on the other hand we are forced to use big data structures in order to deal with the huge number of optimal walks (which is not the case with the underlying graphs).

\begin{figure}
\centering
\begin{adjustbox}{width=0.7\textwidth}
\includegraphics{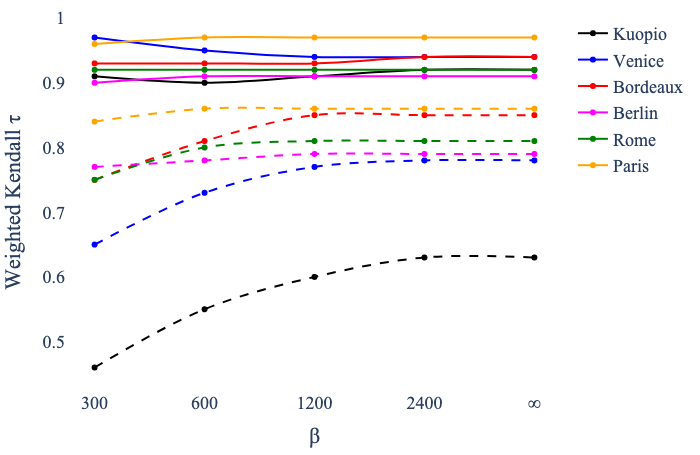}
\end{adjustbox}
\caption{The weighted Kendall $\tau$ between the \sfa\ and the \sfo\ betweenness rankings (solid) and between the \sfo\ betweenness and the betweenness rankings (dashed), for each public transport network and for $\beta=300,600,1200,2400,\infty$.}
\label{fig:wktbrandes}
\end{figure}

Since computing the betweenness of the underlying graph is significantly faster than our algorithm, it is worth determining the correlation between this betweenness and the \sfa\ and the \sfo\ betweenness. As it is shown in Figure~\ref{fig:wktbrandes}, the weighted Kendall $\tau$ between the \sfa\ and the \sfo\ betweenness rankings is very high for all values of $\beta$. On the contrary, the weighted Kendall $\tau$ between the \sfo\ betweenness and the betweenness rankings is significantly lower, especially when $\beta=300$ or $\beta=600$, which correspond to five and ten minutes of waiting time, respectively (similar results hold for the \sfa\ betweenness). These results thus suggest that the betweenness of the underlying graph cannot be used as a proxy of the \sfa\ and the \sfo\ betweenness, whenever waiting constraints have to be satisfied. This contrasts with what has been shown in~\cite{Becker2023}, where the authors found high correlations between the \sh\ betweenness with $\beta=\infty$ (that is, when no waiting constraint are used) and the betweenness of the underlying graph.

\section{Further research}

As we follow the framework of~\cite{BrunelliV2022}, our algorithm can support (with the same complexity and with exact computation) more general waiting constraints (such as minimum waiting time and maximum waiting time varying for each node), and edge weights/costs with appropriate cost and target-cost structures. From a theoretical point of view, it would be interesting to look for conditional lower bound on the computation of the betweenness in the restless case. From an experimental point of view, instead, it would be worth enriching the temporal graph model and appropriately modify our algorithms  in order to take into account the number of transfers in public transport networks, and, thus, focusing on the betweenness values based on the walks with the minimum number of transfers. Finally, we think it would be interesting to experimentally compare our algorithm with more recent papers that propose new exact~\cite{Zhang2024} and approximation~\cite{Cruciani2024} algorithms for computing the temporal betweenness in the non-restless case.

\part*{Appendix}
\section{Proofs of Section~\ref{sec:preliminaries}}
\label{sec:proofs1}

The proofs of Facts~\ref{fact:sfo_prefixoptimality} and~\ref{fact:sfo_toptimalimpliescoptimal} immediately follows from the definition of \sh\ and \sfo\ walks.

\paragraph{Proof of Fact~\ref{fact:betweennessformula}} From the previous definitions, we have that
\begin{align*}
b_{u} & = \sum_{s,t\in V\setminus\{u\}:\chi_{s,t}=1}\frac{\sigma^*_{s,u,t}}{\sigma^*_{s,t}} = \sum_{s\in V\setminus\{u\}}\sum_{t\in V\setminus\{u\}:\chi_{s,t}=1}\frac{\sigma^*_{s,u,t}}{\sigma^*_{s,t}} \\
& = \sum_{s\in V\setminus\{u\}}\sum_{t\in V\setminus\{u\}:\chi_{s,t}=1}\frac{1}{\sigma^*_{s,t}}\sum_{e\in E_{u}^\mathrm{h}}\sigma^{*}_{s,e,t}\\
& = \sum_{s\in V\setminus\{u\}}\sum_{e\in E^\mathrm{h}_u}\sum_{t\in V\setminus\{u\}:\chi_{s,t}=1}\frac{\sigma^*_{s,e,t}}{\sigma^*_{s,t}}\\
& = \sum_{s\in V\setminus\{u\}}\sum_{e\in E^\mathrm{h}_u}\left(\sum_{t\in V:\chi_{s,t}=1}\frac{\sigma^*_{s,e,t}}{\sigma^*_{s,t}}-\chi_{s,u}\frac{\sigma^*_{s,e,u}}{\sigma^*_{s,u}}\right) \\
& = \sum_{s\in V\setminus\{u\}}\sum_{e\in E^\mathrm{h}_u}\left(b_{s,e}-\chi_{s,u}\frac{\sigma^*_{s,e,u}}{\sigma^*_{s,u}}\right) = \\
& = \sum_{s\in V\setminus\{u\}}\left(\sum_{e\in E^\mathrm{h}_u}b_{s,e}-\chi_{s,u}\sum_{e\in E^\mathrm{h}_u}\frac{\sigma^*_{s,e,u}}{\sigma^*_{s,u}}\right) \\
& = \sum_{s\in V\setminus\{u\}}\left(\left(\sum_{e\in E^\mathrm{h}_u}b_{s,e}\right)-\chi_{s,u}\right),
\end{align*}
where the last equality is due to the fact that $\sigma^*_{s,u}=\sum_{e\in E^\mathrm{h}_u}\sigma^*_{s,e,u}$. The fact thus follows.\qed

\paragraph{Proof of Fact~\ref{fact:sfo_prefixtimessuffix}}

Since the $st$-walks in $\mathcal{W}_{s,e,t}$ are \sfo, from Facts~\ref{fact:sfo_toptimalimpliescoptimal} and~\ref{fact:sfo_prefixoptimality} it follows that any $W\in \mathcal{W}_{s,e,t}$, whose last temporal edge is $f$, is a shortest $sf$-walk, and that the prefix of $W$ up to the temporal edge $e$ is a shortest $se$-walk. As a consequence, the total number of distinct prefixes of walks in $\mathcal{W}_{s,e,t}$ is equal to $\sigma_{s,e}$ and each such prefix has the same length, say $c$, which is the length of a \sh\ $se$-walk. Let $W_{1}$ be the prefix of a walk $W\in \mathcal{W}_{s,e,t}$ up to the temporal edge $e$, $X_2=\langle e_{1}e_{2}\cdots e_{k}\rangle$ be the suffix of a walk $X\in \mathcal{W}_{s,e,t}$ following the temporal edge $e$, and $Y=W_{1}.e_{1}.e_{2}.\cdots.e_{k}$. Then $\gamma(Y)=\gamma(X)$ since all prefixes have length $c$. Moreover, both $X$ and $Y$ ends with the temporal edge $e_{k}$. Hence, $\tcf(e_{k},\gamma(X))=^{F}\tcf(e_{k},\gamma(Y))$. Since $X\in \mathcal{W}_{s,e,t}$, then also $Y\in \mathcal{W}_{s,e,t}$. This implies that we can combine any prefix with any suffix and obtain a \sfo\ $st$-walk. We thus have $\sigma^*_{s,e,t}=\sigma_{s,e}\cdot\theta_{s,e,t}$ and the fact follows.\qed

\paragraph{Proof of Lemma~\ref{lem:backwardlemma}}

From Fact~\ref{fact:sfo_prefixtimessuffix} it follows that
\begin{equation}\label{eq:betweenness}
b_{s,e} = \sum_{t\in V:\chi_{s,t}=1}\frac{\sigma^{*}_{s,e,t}}{\sigma^{*}_{s,t}} = \sigma_{s,e}\sum_{t\in V:\chi_{s,t}=1}\frac{\theta_{s,e,t}}{\sigma^{*}_{s,t}}
\end{equation}
By distinguishing the case in which $t=v$, we have 
\begin{eqnarray*}
b_{s,e} & = & \sigma_{s,e}\sum_{t\in V:\chi_{s,t}=1}\frac{\theta_{s,e,t}}{\sigma^{*}_{s,t}} = \sigma_{s,e}\left(\sum_{t\in V\setminus\{v\}:\chi_{s,t}=1}\frac{\theta_{s,e,t}}{\sigma^{*}_{s,t}}+\frac{\theta_{s,e,v}}{\sigma^{*}_{s,v}}\right)\\
 & = & \sigma_{s,e}\left(\sum_{t\in V\setminus\{v\}:\chi_{s,t}=1}\frac{\theta_{s,e,t}}{\sigma^{*}_{s,t}} + \frac{\hat{\theta}_{s,e,v}}{\sigma^{*}_{s,v}}\right)+\left\{\begin{array}{ll}\frac{\sigma^{*}_{s,e}}{\sigma^{*}_{s,v}} & \mbox{if $\sigma^{*}_{s,e}>0$,}\\
0 & \mbox{otherwise,}\end{array}\right.\\
\end{eqnarray*}
where
\[
\hat{\theta}_{s,e,t} = \left\{\begin{array}{ll}
\theta_{s,e,t} & \mbox{if $v\neq t\vee\sigma^{*}_{s,e}=0$,}\\
\theta_{s,e,t}-1 & \mbox{otherwise,}\end{array}\right.
\]
and the second inequality is due to the fact that $\sigma_{s,e}=\sigma^{*}_{s,e}$ whenever $\sigma^{*}_{s,e}>0$. We hence have that

\begin{eqnarray*}
b_{s,e}  = \sigma_{s,e}\sum_{t\in V:\chi_{s,t}=1}\frac{\hat{\theta}_{s,e,t}}{\sigma^{*}_{s,t}}+\left\{\begin{array}{ll}\frac{\sigma^{*}_{s,e}}{\sigma^{*}_{s,v}} & \mbox{if $\sigma^{*}_{s,e}>0$,}\\
0 & \mbox{otherwise.}\end{array}\right.   
\end{eqnarray*}
Note that $\hat{\theta}_{s,e,t}$ counts only non-empty suffixes of $\tcost$-optimal $st$-walks containing $e$ (which exist only if $v=t$ and $\sigma^{*}_{s,e}>0$). In such walks the first temporal edge $f$ of the non-empty suffix is a successor of $e$. This implies 
\[
\hat{\theta}_{s,e,t}=\sum_{f\in\mathtt{succ}_{s,e,t}}\theta_{s,f,t},
\]
which yields:

\begin{eqnarray*}
b_{s,e} & = & \sigma_{s,e}\sum_{t\in V:\chi_{s,t}=1}\frac{\sum_{f\in\mathtt{succ}_{s,e,t}}\theta_{s,f,t}}{\sigma^{*}_{s,t}}+\left\{\begin{array}{ll}\frac{\sigma^{*}_{s,e}}{\sigma^{*}_{s,v}} & \mbox{if $\sigma^{*}_{s,e}>0$,}\\
0 & \mbox{otherwise}\end{array}\right.\\
& = & \sigma_{s,e}\sum_{f\in\mathtt{succ}_{s,e}}\sum_{t\in V:\chi_{s,t}=1}\frac{\theta_{s,f,t}}{\sigma^{*}_{s,t}}+\left\{\begin{array}{ll}\frac{\sigma^{*}_{s,e}}{\sigma^{*}_{s,v}} & \mbox{if $\sigma^{*}_{s,e}>0$,}\\
0 & \mbox{otherwise,}\end{array}\right.\\
& = & \sigma_{s,e}\sum_{f\in\mathtt{succ}_{s,e}}\frac{b_{s,f}}{\sigma_{s,f}}+\left\{\begin{array}{ll}\frac{\sigma^{*}_{s,e}}{\sigma^{*}_{s,v}} & \mbox{if $\sigma^{*}_{s,e}>0$,}\\
0 & \mbox{otherwise,}\end{array}\right.
\end{eqnarray*}
where the second equality is due to the fact the $f\in\mathtt{succ}_{s,e}$ if and only there exists $t\in V\setminus\{s\}$ such that $f\in\mathtt{succ}_{s,e,t}$, and the last equality follows from Equation~\ref{eq:betweenness}. The lemma is thus proved.\qed

\section{An iteration of Algorithm~\ref{alg:sfononrestless}}
\label{sec:nrtsfob}

In Figure~\ref{fig:nrtsfob} (where we use over-lined integers to denote the position of a temporal edge in $E^{\mathrm{arr}}$), we show the iteration of the forward phase of Algorithm~\ref{alg:sfononrestless}, in which the edge $e_{\iarr{k}}=(u,v,\tau,\lambda)$ is analysed. The gray temporal edges in $E^{\mathrm{arr}}$ are the temporal edges already scanned, while the gray temporal edges in $E^{\mathrm{dep}}_\mathrm{node}[u]$ and in $E^{\mathrm{dep}}_\mathrm{node}[v]$ are the temporal edges which have been finalised. The figure shows these two lists at the beginning of the iteration (upper copy) and at the end of the iteration (lower copy): between the two copies, the finalisation of the involved temporal edges is shown.

\begin{figure*}[t]
\begin{adjustbox}{width=\textwidth}
\begin{tikzpicture}
\begin{scope}[yshift=0.5cm]
\draw (4.5,1.5) +(-.5,-.5) node {$E^{\mathrm{arr}}$};
\draw[fill=lightgray] (1,0) +(-1,0) rectangle ++(0,.5);
\draw (1,0.5) +(-.5,-.25) node {$e_{\iarr{1}}$};
\draw[fill=lightgray] (2,0) +(-1,0) rectangle ++(0,.5);
\draw (2,0.5) +(-.5,-.25) node {$e_{\iarr{2}}$};
\draw[fill=lightgray] (3,0) +(-1,0) rectangle ++(0,.5);
\draw (3,0.5) +(-.5,-.25) node {$\cdots$};
\draw[fill=lightgray] (4,0) +(-1,0) rectangle ++(0,.5);
\draw (4,0.5) +(-.5,-.25) node {$e_{\iarr{k-1}}$};
\draw (5,0) +(-1,0) rectangle ++(0,.5);
\draw (5,0.5) +(-.5,-.25) node {$e_{\iarr{k}}$};
\draw (6,0) +(-1,0) rectangle ++(0,.5);
\draw (6,0.5) +(-.5,-.25) node {$e_{\iarr{k+1}}$};
\draw (7,0) +(-1,0) rectangle ++(0,.5);
\draw (7,0.5) +(-.5,-.25) node {$\cdots$};
\draw (8,0) +(-1,0) rectangle ++(0,.5);
\draw (8,0.5) +(-.5,-.25) node {$e_{\iarr{M}}$};
\end{scope}
\begin{scope}[xshift=-4cm,yshift=-1.5cm]
\draw[fill=lightgray] (0,0) rectangle ++(1,.5);
\draw (0.5,0.25) node {$\iarr{j_{1}}$};
\draw[fill=lightgray] (1,0) rectangle ++(1,.5);
\draw (1.5,0.25) node {$\iarr{j_{2}}$};
\draw[fill=lightgray] (2,0) rectangle ++(0.6,.5);
\draw (2.3,0.25) node {$\cdots$};
\draw (2.6,0) rectangle ++(1,.5);
\draw (3.1,0.25) node {$\iarr{j_{l[u]}}$};
\draw (3.6,0) rectangle ++(0.6,.5);
\draw (3.9,0.25) node {$\cdots$};
\draw (4.2,0) rectangle ++(1,.5);
\draw (4.7,0.25) node {$\iarr{j_{i}}$};
\draw (5.2,0) rectangle ++(1,.5);
\draw (5.7,0.25) node {$\iarr{j_{i+1}}$};
\draw (6.2,0) rectangle ++(0.6,.5);
\draw (6.5,0.25) node {$\cdots$};
\draw (6.8,0) rectangle ++(1,.5);
\draw (7.3,0.25) node {$\iarr{j_{d[u]}}$};
\draw (0.65,1.0) node {$E^{\mathrm{dep}}_{\mathrm{node}}[u]$};
\end{scope}
\begin{scope}[xshift=-4cm,yshift=-3cm]
\draw[fill=lightgray] (0,0) rectangle ++(1,.5);
\draw (0.5,0.25) node {$\iarr{j_{1}}$};
\draw[fill=lightgray] (1,0) rectangle ++(1,.5);
\draw (1.5,0.25) node {$\iarr{j_{2}}$};
\draw[fill=lightgray] (2,0) rectangle ++(2.2,.5);
\draw (3.1,0.25) node {$\cdots$};
\draw[fill=lightgray] (4.2,0) rectangle ++(1,.5);
\draw (4.7,0.25) node {$\iarr{j_{i}}$};
\draw (5.2,0) rectangle ++(1,.5);
\draw (5.7,0.25) node {$\iarr{j_{l[u]}}$};
\draw (6.2,0) rectangle ++(0.6,.5);
\draw (6.5,0.25) node {$\cdots$};
\draw (6.8,0) rectangle ++(1,.5);
\draw (7.3,0.25) node {$\iarr{j_{d[u]}}$};
\draw (3.85,1.0) node{$\forall f\in[l[u],i](C[\iarr{j_{f}}]\leftarrow c[u]+1, \Sigma[\iarr{j_{f}}]:=\sigma[u])$};
\end{scope}
\begin{scope}[xshift=+4cm,yshift=-1.5cm]
\draw[fill=lightgray] (0,0) rectangle ++(1,0.5);
\draw (0.5,0.25) node {$\iarr{h_{1}}$};
\draw[fill=lightgray] (1,0) rectangle ++(1,.5);
\draw (1.5,0.25) node {$\iarr{h_{2}}$};
\draw[fill=lightgray] (2,0) rectangle ++(0.6,.5);
\draw (2.3,0.25) node {$\cdots$};
\draw (2.6,0) rectangle ++(1,.5);
\draw (3.1,0.25) node {$\iarr{h_{l[v]}}$};
\draw (3.6,0) rectangle ++(0.6,.5);
\draw (3.9,0.25) node {$\cdots$};
\draw (4.2,0) rectangle ++(1,.5);
\draw (4.7,0.25) node {$\iarr{h_{a-1}}$};
\draw (5.2,0) rectangle ++(1,.5);
\draw (5.7,0.25) node {$\iarr{h_{a}}$};
\draw (6.2,0) rectangle ++(0.6,.5);
\draw (6.5,0.25) node {$\cdots$};
\draw (6.8,0) rectangle ++(1,.5);
\draw (7.3,0.25) node {$\iarr{h_{d[v]}}$};
\draw (7.2,1.0) node {$E^{\mathrm{dep}}_{\mathrm{node}}[v]$};
\end{scope}
\begin{scope}[xshift=+4cm,yshift=-3cm]
\draw (3.85,1.0) node{$\forall f\in[l[v],a-1](C[\iarr{j_{f}}]\leftarrow c[v]+1, \Sigma[\iarr{j_{f}}]\leftarrow\sigma[v])$};
\draw[fill=lightgray] (0,0) rectangle ++(1,0.5);
\draw (0.5,0.25) node {$\iarr{h_{1}}$};
\draw[fill=lightgray] (1,0) rectangle ++(1,.5);
\draw (1.5,0.25) node {$\iarr{h_{2}}$};
\draw[fill=lightgray] (2,0) rectangle ++(2.2,.5);
\draw (3.1,0.25) node {$\cdots$};
\draw[fill=lightgray] (4.2,0) rectangle ++(1,.5);
\draw (4.7,0.25) node {$\iarr{h_{a-1}}$};
\draw (5.2,0) rectangle ++(1,.5);
\draw (5.7,0.25) node {$\iarr{h_{l[v]}}$};
\draw (6.2,0) rectangle ++(0.6,.5);
\draw (6.5,0.25) node {$\cdots$};
\draw (6.8,0) rectangle ++(1,.5);
\draw (7.3,0.25) node {$\iarr{h_{d[v]}}$};
\end{scope}
\draw[->] (4.5,0.5) -- (4.5,-0.5) -- node[above] {$\iarr{k}=\iarr{j_{i}}$} (0.75,-0.5) -- (0.75,-1.00);
\draw[->] (4.5,0.0) -- (4.5,-0.5) -- node[above] {$\mathrm{dep}(e_{\iarr{h_{a}}})\geq\mathrm{arr}(e_{\iarr{k}})$} (9.75,-0.5) -- (9.75,-1.00);
\end{tikzpicture}
\end{adjustbox}
\caption{}
\label{fig:nrtsfob}
\end{figure*}

\section{Proof of Theorem~\ref{thm:sfo_alg1analysis}}
\label{sec:proofs2}

As it can be seen in Algorithm~\ref{alg:sfononrestless}, our algorithm executes a constant number of cycles on all nodes and on all temporal edges. More precisely, both the forward phase and the backward phase scan once each edge in $E^{\mathrm{arr}}$ and each index in $E^{\mathrm{dep}}_\mathrm{node}$. Hence, overall the time complexity of the algorithm is linear in the number of temporal edges. For what concerns the correctness of the algorithm, let us consider the forward and the backward phase separately.

\begin{itemize}
\item During the forward phase, the following invariant is maintained after scanning $k$th edge in $E^{\mathrm{arr}}$: $\forall v\in V$, $c[v]$ is the shortest length of an $sv$-walk in $G_k$, $\sigma[v]$ is the number of shortest $sv$-walks in $G_{k}$ (with length $c[v]$), and all edges in $E^{\mathrm{dep}}_\mathrm{node}[v][l[v]:]$ extend such walks. The proof follows easily by induction on $k$ and the fact that any walk in $G_k$ which is not in $G_{k-1}$ must end with $e_k$. This comes from the fact that any temporal walk is strict since we assume positive travel times, and its temporal edges must appear in order in $E^{\mathrm{arr}}$ which is sorted by arrival times.

\item During the backward phase, the following invariant is maintained after scanning backward the $k$th edge in $E^{\mathrm{arr}}$: $\delta[v]=\sum_{f\in E_{\mathrm{node}}^{\mathrm{dep}}[v][l[v]:r]}\frac{b_{s,f}}{\sigma_{s,f}}$, where $r$ denotes the last index of an edge $f\in E_{\mathrm{node}}^{\mathrm{dep}}[v]$ such that a shortest $sf$-walk in $G$ can be obtained by extending a walk in $G_k$. This follows from Lemma~\ref{lem:backwardlemma} and an easy induction on $k$ as the index $r$ can only decrease as $k$ decreases. Similarly as in the forward phase, we use the fact that $E^{\mathrm{arr}}$ is a topological ordering for the successor relation: if $f$ is successor of $e$, then $f$ appears after $e$ in $E^{\mathrm{arr}}$.
\end{itemize}
The theorem is then proved.\qed

\section{Execution of Algorithm~\ref{alg:sfononrestless}}
\label{sec:exampleexecution}
Let us consider the temporal graph of Figure~\ref{fig:patg2}, where the four lists $E^{\mathrm{arr}}$, $E^{\mathrm{dep}}$, $E^{\mathrm{dep}}_\mathrm{node}$, and $E^{\mathrm{arr}}_{\mathrm{dep}}$ are also depicted. In Table~\ref{tab:alg2exefor2}, the execution of the forward phase (lines~\ref{alg2:forwardstart}-\ref{alg2:forwardend}), with source node $v_{1}$, is represented by showing the values of the main variables of the algorithm at the end of the analysis of each temporal edge in $E^{\mathrm{arr}}$. Let us analyse two specific consecutive iterations of this phase.

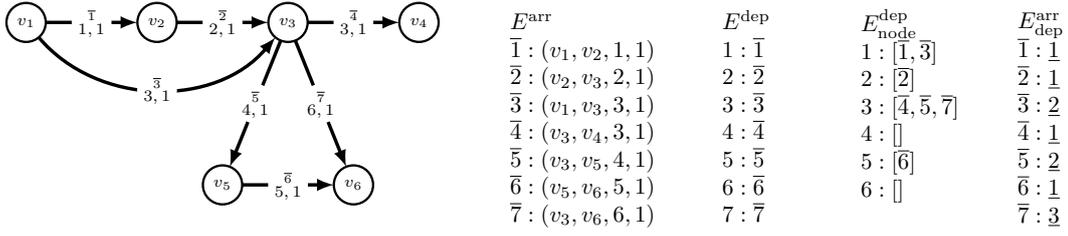
\begin{figure*}[ht]
\centering
\begin{adjustbox}{width=\textwidth}
    \SetVertexStyle[FillColor=white]
\SetEdgeStyle[Color=black]
\begin{tikzpicture}[x=2cm,y=2cm]
\Vertex[x=0,y=0,label=$v_{1}$]{v1}
\Vertex[x=2,y=0,label=$v_{2}$]{v2}
\Vertex[x=4,y=0,label=$v_{3}$]{v3}
\Vertex[x=6,y=0,label=$v_{4}$]{v4}
\Vertex[x=3,y=-2.5,label=$v_{5}$]{v5}
\Vertex[x=5,y=-2.5,label=$v_{6}$]{v6}
\Edge[Direct,label={$\stackrel{\overline{1}}{1,1}$}](v1)(v2)
\Edge[Direct,label={$\stackrel{\overline{2}}{2,1}$}](v2)(v3)
\Edge[Direct,label={$\stackrel{\overline{3}}{3,1}$},bend=-50](v1)(v3)
\Edge[Direct,label={$\stackrel{\overline{4}}{3,1}$}](v3)(v4)
\Edge[Direct,label={$\stackrel{\overline{5}}{4,1}$}](v3)(v5)
\Edge[Direct,label={$\stackrel{\overline{6}}{5,1}$}](v5)(v6)
\Edge[Direct,label={$\stackrel{\overline{7}}{6,1}$}](v3)(v6)
\node at (4.25,-0.75) {
\begin{tabular}{l}
$E^{\mathrm{arr}}$\\
$\overline{1}: (v_{1},v_{2},1,1)$\\
$\overline{2}: (v_{2},v_{3},2,1)$\\
$\overline{3}: (v_{1},v_{3},3,1)$\\
$\overline{4}: (v_{3},v_{4},3,1)$\\
$\overline{5}: (v_{3},v_{5},4,1)$\\
$\overline{6}: (v_{5},v_{6},5,1)$\\
$\overline{7}: (v_{3},v_{6},6,1)$\\
\end{tabular}
};
\node at (5.5,-0.75) {
\begin{tabular}{l}
$E^{\mathrm{dep}}$\\
$1: \overline{1}$\\
$2: \overline{2}$\\
$3: \overline{3}$\\
$4: \overline{4}$\\
$5: \overline{5}$\\
$6: \overline{6}$\\
$7: \overline{7}$\\
\end{tabular}
};
\node at (6.75,-0.65) {
\begin{tabular}{l}
$E_\mathrm{node}^{\mathrm{dep}}$\\
$1: [\overline{1},\overline{3}]$\\
$2: [\overline{2}]$\\
$3: [\overline{4},\overline{5},\overline{7}]$\\
$4: []$\\
$5: [\overline{6}]$\\
$6: []$\\
\end{tabular}
};
\node at (7.75,-0.75) {
\begin{tabular}{l}
$E^{\mathrm{arr}}_{\mathrm{dep}}$\\
$\overline{1}: \underline{1}$\\
$\overline{2}: \underline{1}$\\
$\overline{3}: \underline{2}$\\
$\overline{4}: \underline{1}$\\
$\overline{5}: \underline{2}$\\
$\overline{6}: \underline{1}$\\
$\overline{7}: \underline{3}$\\
\end{tabular}
};
\end{tikzpicture}
\end{adjustbox}
\caption{An example of a temporal graph, where $n=6$, $M=7$, $T=6$.}
\label{fig:patg2}
\end{figure*}

\begin{description}
\item[The temporal edge $e=(v_{1},v_{3},3,1)$ is scanned.] This is the second and last temporal edge in $E^{\mathrm{dep}}_\mathrm{node}[v_{1}]$. Since $l[v_{1}]$ has been set equal to $\idep{2}$ while scanning the first temporal edge, all temporal edges in $E^{\mathrm{dep}}_\mathrm{node}[v_{1}]$ are now finalised. However, $c[v_{1}]=\infty$, so that the effect of the finalisation is just to set $l[v_{1}]=\idep{3}$ (indeed, there are no more temporal edges with the tail equal to $v_{1}$). The tail of $e$ is the source node: hence, $C[e]$ and $\Sigma[e]$ are set equal to $1$ because $\langle e\rangle$ is, indeed, the unique shortest $se$-walk. Since, now, $C[e]=1<2=c[v_{3}]$ and since the first temporal edge in $E^{\mathrm{dep}}_\mathrm{node}[v_{3}]$ that can extend $e$ is the second one (that is, $e_{\iarr{5}}$), the first edge of $E^{\mathrm{dep}}_\mathrm{node}[v_{3}]$ (that is, $e_{\iarr{4}}$) is finalised, by setting $C[e_{\iarr{4}}]=c[v_{3}]+1=3$ and $\Sigma[e_{\iarr{4}}]=\sigma[v_{3}]=1$: as we will see, this will allow the algorithm to assign the correct value to $c[v_{4}]$ in the next iteration. Finally, since $C[e]<c[v_{3}]$, this implies that the algorithm has found a shorter $v_{1}v_{3}$-walk: for this reason, $c[v_{3}]$ is set equal to $C[e]=1$ and $\sigma[v_{3}]$ is set equal to $\Sigma[e]=1$. Finally, the algorithm sets $L[e]$ equal to $\idep{2}$, in order to remember, during the backward phase, that the first temporal edge in $E^{\mathrm{dep}}_\mathrm{node}[v_{3}]$, that can extend $e$ and result in a \sh\ walk, is the second one.

\begin{table*}[ht]
\centering
\begin{adjustbox}{width=\textwidth}
    \begin{tabular}{||c||c|c|c||c|c|c|c||c|c|c|c||}
\hline
$e=(u,v,\tau\lambda)$ & $L[e]$ & $C[e]$ & $\Sigma[e]$ & $E_\mathrm{node}^{\mathrm{dep}}[u]$ & $l[u]$ & $c[u]$ & $\sigma[u]$ & $E_\mathrm{node}^{\mathrm{dep}}[v]$ & $l[v]$ & $c[v]$ & $\sigma[v]$\\
\hline\hline
$\iarr{1}=(v_{1}, v_{2}, 1, 1)$ & $\idep{1}$ & $1$ & $1$ & $[\iarr{1}, \iarr{3}]$ & $\idep{2}$ & $\infty$ & $0$ & $[\iarr{2}]$ & $\idep{1}$ & $1$ & $1$\\
$\iarr{2}=(v_{2}, v_{3}, 2, 1)$ & $\idep{1}$ & $2$ & $1$ & $[\iarr{2}]$ & $\idep{2}$ & $1$ & $1$ & $[\iarr{4}, \iarr{5}, \iarr{7}]$ & $\idep{1}$ & $2$ & $1$\\
$\iarr{3}=(v_{1}, v_{3}, 3, 1)$ & $\idep{2}$ & $1$ & $1$ & $[\iarr{1}, \iarr{3}]$ & $\idep{3}$ & $\infty$ & $0$ & $[\iarr{4}, \iarr{5}, \iarr{7}]$ & $\idep{2}$ & $1$ & $1$\\
$\iarr{4}=(v_{3}, v_{4}, 3, 1)$ & $\idep{1}$ & $3$ & $1$ & $[\iarr{4}, \iarr{5}, \iarr{7}]$ & $\idep{2}$ & $1$ & $1$ & $[]$ & $\idep{1}$ & $3$ & $1$\\
$\iarr{5}=(v_{3}, v_{5}, 4, 1)$ & $\idep{1}$ & $2$ & $1$ & $[\iarr{4}, \iarr{5}, \iarr{7}]$ & $\idep{3}$ & $1$ & $1$ & $[\iarr{6}]$ & $\idep{1}$ & $2$ & $1$\\
$\iarr{6}=(v_{5}, v_{6}, 5, 1)$ & $\idep{1}$ & $3$ & $1$ & $[\iarr{6}]$ & $\idep{2}$ & $2$ & $1$ & $[]$ & $\idep{1}$ & $3$ & $1$\\
$\iarr{7}=(v_{3}, v_{6}, 6, 1)$ & $\idep{1}$ & $2$ & $1$ & $[\iarr{4}, \iarr{5}, \iarr{7}]$ & $\idep{4}$ & $1$ & $1$ & $[]$ & $\idep{1}$ & $2$ & $1$\\
\hline
\end{tabular}
\end{adjustbox}
\caption{The evolution of the main variables of Algorithm~\ref{alg:sfononrestless} during the forward phase.}
\label{tab:alg2exefor2}
\end{table*}

\begin{table*}[ht]
\centering
\begin{tabular}{||c||c|c|c|c|c||c|c|c|c|c||}
\hline
$e=(u,v,\tau\lambda)$ & $L[e]$ & $C[e]$ & $\Sigma[e]$ & $\Sigma^{*}[E]$ & $b[e]$ & $E_\mathrm{node}^{\mathrm{dep}}[v]$ & $l[v]$ & $c[v]$ & $\sigma^{*}[v]$ & $\delta[v]$\\
\hline\hline
$\iarr{7}=(v_{3}, v_{6}, 6, 1)$ & \idep{1} & $2$ & $1$ & $0$ & $0.0$ & $[]$ & \idep{1} & $2$ & $1$ & $0.0$\\
$\iarr{6}=(v_{5}, v_{6}, 5, 1)$ & \idep{1} & $3$ & $1$ & $1$ & $1.0$ & $[]$ & \idep{1} & $2$ & $1$ & $0.0$\\
$\iarr{5}=(v_{3}, v_{5}, 4, 1)$ & \idep{1} & $2$ & $1$ & $1$ & $2.0$ & $[\iarr{6}]$ & \idep{1} & $2$ & $1$ & $1.0$\\
$\iarr{4}=(v_{3}, v_{4}, 3, 1)$ & \idep{1} & $3$ & $1$ & $1$ & $1.0$ & $[]$ & \idep{1} & $3$ & $1$ & $0.0$\\
$\iarr{3}=(v_{1}, v_{3}, 3, 1)$ & \idep{2} & $1$ & $1$ & $0$ & $2.0$ & $[\iarr{4}, \iarr{5}, \iarr{7}]$ & \idep{2} & $1$ & $1$ & $2.0$\\
$\iarr{2}=(v_{2}, v_{3}, 2, 1)$ & \idep{1} & $2$ & $1$ & $1$ & $2.0$ & $[\iarr{4}, \iarr{5}, \iarr{7}]$ & \idep{1} & $1$ & $1$ & $1.0$\\
$\iarr{1}=(v_{1}, v_{2}, 1, 1)$ & \idep{1} & $1$ & $1$ & $1$ & $3.0$ & $[\iarr{2}]$ & \idep{1} & $1$ & $1$ & $2.0$\\
\hline
\end{tabular}
\caption{The evolution of the main variables of Algorithm~\ref{alg:sfononrestless} during the backward phase.}
\label{tab:alg2exebac2}
\end{table*}

\item[The temporal edge $e=(v_{3},v_{4},3,1)$ is scanned.] This is the first temporal edge in $E^{\mathrm{dep}}_\mathrm{node}[v_{3}]$. Since $l[v_{3}]$ has been set equal to $\idep{2}$ in the previous iteration, no finalisation has to be executed (indeed, the cost $c[v_{3}]=1$ refers to the path $\langle e_{\iarr{3}}\rangle$ which is not extendable by $e$). Recall that, in the previous iteration, the algorithm has set $C[e]$ equal to $3$, which is less than $\infty=c[v_{4}]$. Since $E^{\mathrm{dep}}_\mathrm{node}[v_{4}]$ is empty, no finalisation takes place, but, due to the fact that $C[e]<c[v_{4}]$, $c[v_{4}]$ is set equal to $3$ and $\sigma[v_{4}]$ is set equal to $1$. That is, the algorithm has correctly computed the length of the unique \sh\ $v_{1}v_{4}$-walk. Finally, the algorithm sets $L[e]$ equal to $\idep{1}$: since $\left|E^{\mathrm{dep}}_\mathrm{node}[v_{4}]\right|=0$, this implies that no temporal edge can extend $e$. 
\end{description}

Note how, at the end of the forward phase, the value of $l[v]$ is equal to $\idep{\left|E^{\mathrm{dep}}_\mathrm{node}[v]\right|+1}$, for each node $v$ (since all indices of $E^{\mathrm{dep}}_\mathrm{node}[v]$ have been scanned for finalizing the corresponding temporal edges). This guarantees that, in the following backward phase, all temporal edges in $E^{\mathrm{dep}}_\mathrm{node}[v]$ will be considered. Before starting the backward phase, the algorithm computes, for each node $v$, the values $c^{*}[v]$ (lines~\ref{alg2:pre-backward1}-\ref{alg2:pre-backward2}): to this aim, it just look for the temporal edge with head $v$ such that the pair $(\arr(e),C[e])$ is `lexicographically` minimum (that is, minimum with respect to the $\vartriangleleft$ relation). In lines~\ref{alg2:pre-backward3}-\ref{alg2:pre-backward4}, the algorithm also computes, for each temporal edge $e=(u,v,\tau,\lambda)$, the number $\Sigma^{*}[e]$ of \sfo\ $sv$-walks ending with $e$ and, for each node $v$, the number $\sigma^{*}[v]$ of \sfo\ $sv$-walks (which is the sum of the values $\Sigma^{*}[e]$ for all temporal edges $e$ with head $v$).

In Table~\ref{tab:alg2exebac2}, the execution of the backward phase (lines~\ref{alg2:backwardstart}-\ref{alg2:backwardend}), with source node $v_{1}$, is represented by showing the values of the main variables of the algorithm at the end of the analysis of each temporal edge in $E^{\mathrm{arr}}$ (in reverse order). Let us analyse one specific iteration of this phase, that is, the one corresponding to the temporal edge $e=(v_{2}, v_{3}, 2, 1)$, which is the terminal temporal edge of a \sfo\ $v_{1}v_{3}$-walk and an intermediate temporal edge of the \sfo\ $v_{1}v_{4}$-walk. Note that, at the end of the previous iteration, node $v_{3}$ has accumulated in $\delta[v_{3}]=2$ the two \sfo\ walks from $v_{1}$ to $v_{5}$ and from $v_{1}$ to $v_{6}$. However, these walks do not use the temporal edge $e$: hence, $\delta[v_{3}]$ should not be ``transmitted'' to $e$. That is why, at line~\ref{alg2:zerodelta}, the algorithm checks whether $c^*[v_3]\vartriangleleft (\arr(e),C[e])$. Since this is the case, the value of $\delta[v_{3}]$ is zeroed and the current value of $c^*[v_3]$ is set equal to $(\arr(e),C[e])$. The algorithm then continues, as we already said, by simply applying Lemma~\ref{lem:backwardlemma}. Note that, at each iteration, the algorithm ``moves'' left the right extreme of the sub-list of $E^{\mathrm{dep}}_\mathrm{node}[v]$ to be considered by the following iterations, by assigning to $l[v]$ the value $L[e]$ of the currently scanned temporal edge $e$.

\section{Computing the \sh\ betweenness}
\label{sec:shortestnonrestless}

It is easy to see that Fact~\ref{fact:sfo_toptimalimpliescoptimal} can adapted to the case in which we consider \sh\ walks as follows.

\begin{fact}\label{fact:sh_toptimalimpliescoptimal}
Given a temporal graph ${G}=(V,{E},\beta)$, let $W$ be a \sh\ $st$-walk (for some $s,t\in V$) and let $e$ be the last temporal edge of $W$. Then, $W$ is a \sh\ walk among the $se$-walks.
\end{fact}

Moreover, all the definitions concerning the \sfo\ betweenness can be appropriately adapted to the \sh\ case and all the results proved for the \sfo\ betweenness can be proved also for the \sh\ betweenness. The forward phase of the algorithm for computing the non-restless \sh\ betweenness can then be the same as in the case of the \sfo\ betweenness, as shown in lines~\ref{alg1:forwardstart}-\ref{alg1:forwardend} of Algorithm~\ref{alg:shnonrestless} (note that also the finalisation function is exactly the same). The intermediate phase of the algorithm can be even simplified since the values $c^{*}[v]$ are now integers values representing the length of \sh\ $sv$-walks (we are not interested anymore in the arrival time of the walks).  Finally, the backward phase is exactly the same as in the case of the \sfo\ betweenness apart from the fact that the $\delta$-value of a node is zeroed only if the currently scanned edge with head $v$ is not a terminal edge of a \sh\ $sv$-walk. In particular, at line~\ref{alg1:zerodelta}, the algorithm checks whether, for the currently scanned edge $e$ with head $v$, $C[e]>c[v]$: if this the case, then $e$ is not a terminal edge of a \sh\ $sv$-walk, but, since $L[e]>0$, $e$ is an intermediate edge of a \sh\ $su$-walk, for some other node $u$.

Similarly to what we have done in Appendix~\ref{sec:proofs2}, we can then prove the following result.

\begin{theorem}\label{thm:sh_alg1analysis}
For any temporal graph $G=(V,E,\infty)$ and for any $s\in V$, Algorithm~\ref{alg:shnonrestless} correctly computes the $s$-\sh\ betweenness $b_{s,e}$ in time $O(M)$.
\end{theorem}

\begin{algorithm}[ht]
\small
\Input{$G=(V,E,\infty)$ (represented by $E^{\mathrm{dep}}$ and $E^{\mathrm{arr}}$) and $s\in V$}
\Output{$s$-\sh\ betweenness $b_{s,e}$, for all $e\in E$}

Compute the lists $E^{\mathrm{dep}}_{\mathrm{node}}$ and $E^{\mathrm{arr}}_{\mathrm{dep}};$\\
\lForEach{$v\in V$}{$l[v]:=1;c[v]:=\infty;\sigma[v]:=0;$} 
\lForEach{$e\in E$}{$L[e]:=0;C[e]:=\infty;\Sigma[e]:=0;$}
\ForEach{$e = (u,v,\tau,\lambda)\in E^{\mathrm{arr}}$}{\label{alg1:forwardstart}
    \lIf{$E^{\mathrm{arr}}_{\mathrm{dep}}[e]\geq l[u]$}{$\processcosts(u,E^{\mathrm{arr}}_{\mathrm{dep}}[e]);$}
    \lIf{$u=s$}{$C[e]:=1; \Sigma[e]:=1;$}
    \If{$C[e]\neq\infty$}{
        \If{$C[e]\leq c[v]$}{
            $a=l[v]; D := E^{\mathrm{dep}}_{\mathrm{node}}[v];$\\
            \lWhile{$a\leq|D|\wedge\dep(E^{\mathrm{arr}}[D[a]]) < \tau+\lambda$}{$a:=a+1;$}
            $\processcosts(v,a-1);$\\
            \lIf{$C[e]<c[v]$}{$c[v]:=C[e]; \sigma[v]:=0;$}
            $\sigma[v] := \sigma[v] + \Sigma[e]$; $L[e]:=a;$\\
        }
    }
\label{alg1:forwardend}}
\lForEach{$v\in V$}{$c^{*}[v]:=c[v];\sigma^{*}[v]:=\sigma[v];\delta[v]:=0;$} 
\lForEach{$e\in E$}{$\Sigma^{*}[e]:=0;b[e]:=0;$}
\ForEach{$e = (u,v,\tau,\lambda)\in E^{\mathrm{arr}}$} {\lIf{$C[e]=c^{*}[v]<\infty$}{$\Sigma^{*}[e]:=\Sigma[e];$}}\label{alg1:pre-backward2}
\ForEach{$e = (u,v,\tau,\lambda)\in \mathrm{reverse}(E^{\mathrm{arr}}):L[e]>0$}{\label{alg1:backwardstart}
    \lIf{$c[v]<C[e]$}{$\delta[v]:=0; c[v]:=C[e];$}\label{alg1:zerodelta}
    \lFor{$f\in E_{\mathrm{node}}^{\mathrm{dep}}[v][L[e]:l[v]-1]$}{$\delta[v] := \delta[v] + b[f]/\Sigma[f];$}
    $l[v]:=L[e]; b[e]:=\Sigma[e]\delta[v];$\\
    \lIf{$\Sigma^{*}[e]>0$}{$b[e]:=b[e]+\Sigma^{*}[e]/\sigma^{*}[v];$}
\label{alg1:backwardend}}
\Return $b$ 
\BlankLine
\SetArgSty{textbf}
\Finalize{$u,j$}:\\
\Indp
    \If{$c[u]\neq\infty$}{
        \lForEach{$f\in E^{\mathrm{dep}}_{\mathrm{node}}[l[u],j]$}{$C[f]:=c[u]+1; \Sigma[f]:=\sigma[u];$}}
    $l[u]:=j+1;$
\caption{Compute non-restless \sh\ $b_{s,e}$ for all $e\in E$}
\label{alg:shnonrestless}
\end{algorithm}

\begin{figure*}[ht]
\centering
\begin{tikzpicture}[x=0.75pt,y=0.75pt,yscale=-1,xscale=1,every node/.append style={text height=2ex,text depth=0.5ex}]
\draw   (72.03,29.5) -- (290.3,29.5) -- (290.3,49.7) -- (72.03,49.7) -- cycle ;
\draw    (109.97,29.2) -- (109.87,49.6) ;
\draw    (210.63,29.2) -- (210.53,49.6) ;
\draw  [fill={lightgray}] (72.03,29.5) -- (109.87,29.5) -- (109.87,49.6) -- (72.03,49.6) -- cycle ;
\draw  [color={red},draw opacity=1 ] (110,26) -- (210.58,26) -- (210.58,54) -- (110,54) -- cycle ;
\draw[dashed]    (90.77,29.2) -- (90.67,49.6) ;
\draw[dashed]    (130.37,29.2) -- (130.27,49.6) ;
\draw[dashed]    (189.97,29.6) -- (189.87,50) ;

\draw (91,5) node {$E_\mathrm{node}^{\mathrm{dep}}[v]$};
\draw (160,15) node {\textcolor[rgb]{0.82,0.01,0.11}{$Q_{1}.I$}};
\draw (98,63) node {$p$};
\draw (120,63) node {$l_{1}$};
\draw (200.33,63) node {$r_{1}$};

\begin{scope}[xshift=200]
\draw   (72.03,29.5) -- (290.3,29.5) -- (290.3,49.7) -- (72.03,49.7) -- cycle ;
\draw    (210.63,29.2) -- (210.53,49.6) ;
\draw  [color={red}] (150.6,26) -- (210.58,26) -- (210.58,54) -- (150.6,54) -- cycle ;
\draw  [fill={lightgray}] (72.03,29.6) -- (150.6,29.6) -- (150.6,49.7) -- (72.03,49.7) -- cycle ;
\draw  [color={blue}] (210.58,26) -- (270.27,26) -- (270.27,54) -- (210.58,54) -- cycle ;
\draw  [color={green}] (111.32,26) -- (150.6,26) -- (150.6,54) -- (111.32,54) -- cycle ;\draw[dashed]    (130.1,29.2) -- (130.06,37.07) -- (130,49.6) ;
\draw[dashed]    (170.37,29.53) -- (170.31,40.41) -- (170.27,49.93) ;
\draw[dashed]    (189.97,29.6) -- (189.87,50) ;
\draw[dashed]    (229.97,28.93) -- (229.87,49.33) ;
\draw[dashed]    (250.3,29.6) -- (250.2,50) ;

\draw (91,5) node {$E_\mathrm{node}^{\mathrm{dep}}[v]$};
\draw (180.67,15) node {\textcolor{red}{$Q_{1}.I$}};
\draw (140,63) node {$p$};
\draw (165,63) node {$l_{1}'$};
\draw (200,63) node {$r_{1}$};
\draw (240.33,15) node {\textcolor{blue}{$Q_{2}.I$}};
\draw (220,63) node {$l_{2}$};
\draw (260,63) node {$r_{2}$};
\draw (130,15) node {\textcolor{green}{$F$}};
\end{scope}
\end{tikzpicture}
\caption{Status of the list $E_v^{dep}$ before and after the update in the restless algorithm.}
\label{fig:edepvstatus}
\end{figure*}
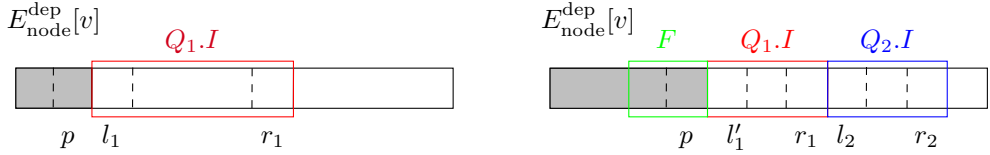

\section{The restless algorithm for \sfo}
\label{sec:sforestless}

To better understand the data structure introduced in Section~\ref{sec:general} and why we need a list of interval quintuples, let us suppose that the list $E^{\mathrm{dep}}_\mathrm{node}[v]$ has been processed in previous iterations up to its index $p$, and that $\intervs{v}$ contains only one interval quintuple $Q_{1}=(l_{1},r_{1},c_{1},\{e_{1}',e_{1}''\},\eta_{1})$ (see the left part of Figure~\ref{fig:edepvstatus}). This means that, for each edge $f\in Q_{1}.I=E^{\mathrm{dep}}_\mathrm{node}[v][l_{1},r_{1}]$, the length of the \sh\ $sv$-walk found so far, which is extended by $f$, is $c_{1}$, that there are $\eta_{1}$ $sv$-walks with length $c_{1}$ that are extended by $f$, and that the edges preceding $f$ in these walks are $e_{1}'$ or $e_{1}''$. Let us suppose that the next temporal edge to be scanned is $e=(u,v,\tau,\lambda)$ and that the minimum length of any $se$-walk is $c$. The edges in $Q_{1}.I$ with departure time earlier than $\arr(e)$ do not extend $e$ nor any other temporal edge that will scanned in the next iterations. Thus they can be finalized, if they have not been already (we will describe the finalisation process below). These edges correspond to those in the interval $F$ in the right part of of Figure~\ref{fig:edepvstatus}, and have index in $E^{\mathrm{dep}}_\mathrm{node}[v]$ from $l_{1}$ to $l_{1}'-1$, where  $l_{1}'$ is the index of the first temporal edge in $Q_{1}.I$ whose departure time is greater than or equal to $\arr(e)$. This means that $Q_{1}$ is now updated to $Q_{1}=(l_{1}',r_{1},c_{1},\{e_{1}',e_{1}''\},\eta_{1})$, and $p$ is updated to $l_{1}'-1$. Now we also have to consider the temporal edges with index greater than $r_{1}$. These edges were not reachable before this iteration. However, those of them that extends $e$ are now reachable and do extend an $se$-walk with length $c$ (these edges are those that have index between $l_{2}$ and $r_{2}$ in the figure). If $c$ is equal to $c_{1}$, then the minimum length of a walk from $s$ that reaches all the temporal edges with index between $l_{1}'$ and $r_{2}$ would still be $c_{1}$ but $e$ has to be added to set of predecessors and the number $\Sigma[e]$ of $se$-walk with length $c$ has to be added to $\eta_{1}$: that is, $Q_{1}$ is updated to $Q_{1}=(l_{1}',r_{2},c_{1},\{e_{1}',e_{1}'',e\},\eta_{1}+\Sigma[e])$. If, instead, $c$ is less than $c_{1}$, then the minimum length of a walk from $s$ that reaches all the temporal edges with index between $l_{1}'$ and $r_{2}$ would be $c$ and these temporal edges can be grouped into a single interval: that is, $Q_{1}$ is updated to $Q_{1}=(l_{1}',r_{2},c,\{e\},\Sigma[e])$. Finally, if $c$ is greater than $c_{1}$ (which is the case represented in the figure), we need to add an interval quintuple $Q_{2}=(l_{2},r_{2},c,\{e\},\Sigma[e])$, where $l_{2}=r_{1}+1$. We also store in the variables $L[e]$ and $R[e]$ the left and right extreme of the last interval quintuple in $\intervs{v}$ (either $Q_1$ in the two first cases, or $Q_2$ in the third case): these values correspond to the interval of edges having $e$ as predecessor and will be used in the backward phase. When $\intervs{v}$ contains more quintuples, we can proceed similarly and need to inspect only quintuples at the beginning or at the end of $\intervs{v}$.

The finalisation of an edge is also more complicated than in the case of the non-restless \sh\ and \sfo\  case. Suppose that we want to finalise all edges in $E^{\mathrm{dep}}_\mathrm{node}[v]$ from the position $l[v]$ to the position $j$. Since these positions are now partitioned into intervals corresponding to the interval quintuples in $\intervs{v}$, we scan these interval quintuples, starting from the first one, until we find an interval quintuple $Q$ such that $Q.l>j$. For each scanned interval quintuple $(l,r,c,P,\eta)$ and for each edge $e\in P$ such that $R[e]\leq j$, we can finalize all edges $f$ in $E^{\mathrm{dep}}_\mathrm{node}[v][l:R[e]]$, that is, we can set the minimum length $C[f]$ of a $sf$-walk equal to the minimum length $c$ of a \sh\ walk which can be extended by $f$ plus one, and the number of these $sf$-walks equal to the number of \sh\ walks that can be extended by $f$ (that is, we set $\Sigma[f]=\eta$). Note that, once all edges $f$ have been finalised, the number of \sh\ walks that can now be extended has to be reduced by subtracting to it the number of \sh\ $se$-walks (that is, we set $\eta=\eta-\Sigma[e]$), since all these latter walks have now been indeed extended. If a scanned quintuple has been emptied (that is, if $j\geq r$), then we can remove it from $\intervs{v}$ and continue with the next one. Otherwise,  we can update it by setting its left extreme equal to $j+1$. Finally, at the end of the scanning of $\intervs{v}$, we can also set $l[v]$ equal to $j+1$ (as we already did in the non-restless cases).
 
Once computed (similarly to the non-restless case), for any node $v$, the earliest arrival time among all $sv$-walks together with the minimum length of such \fo\ walks, the number of such \sfo\ $sv$-walks, and, for any temporal edge $e$, the number of \sfo\ $se$-walks, the backward phase of the general algorithm applies, once again, Lemma~\ref{lem:backwardlemma}. As in the case of the non-restless case, the contribution of the successors of a currently scanned temporal edge $e$ is accumulated in the head of $e$ (more precisely, in the variable $\delta[v]$). In order to avoid to count more than once each contribution, the contributions of the temporal edges in $E^{\mathrm{dep}}_\mathrm{node}[v]$ on the right of position $R[e]$ is subtracted to $\delta[v]$. Subsequently, the contributions of the temporal edges in $E^{\mathrm{dep}}_\mathrm{node}[v]$ from position $L[e]$ and position $R[e]$ which are not already included are added to $\delta[v]$: the new value of $\delta[v]$ can then be used to compute the temporal betweenness of $e$, according to Lemma~\ref{lem:backwardlemma}.

The pseudo-code of the restless algorithm for the \sfo\ betweenness is shown in Algorithm~\ref{alg:sfo_restless}. It has to be observed that the actual code includes, in the backward phase, a control structure which allows us to deal with numerical approximation problems. This structure is not needed if big integer and big rational data structures are used, which, on the other hand, may slower the execution of the algorithm by a factor between two and four.

\begin{algorithm*}[t]
\small
\Input{$G=(V,E,\beta)$ (represented by $E^{\mathrm{dep}}$ and $E^{\mathrm{arr}}$) and $s\in V$}
\Output{$s$-\sfo\ betweenness $b_{s,e}$, for all $e\in E$}

Compute the lists $E^{\mathrm{dep}}_{\mathrm{node}}$ and $E^{\mathrm{arr}}_{\mathrm{dep}};$\\
\lForEach{$v\in V$}{$l[v]:=1;r[v]:=0;\intervs{}[v]:=\emptyset;$} 
\lForEach{$e\in E$}{$L[e]:=\left|E^{arr}\right|+1;R[e]:=0;C[e]:=\infty;\Sigma[e]:=0;$}
\ForEach{$e = (u,v,\tau,\lambda)\in E^{\mathrm{arr}}$}{\label{alg3:forwardstart}
    \lIf{$E^{\mathrm{arr}}_{\mathrm{dep}}[e]\geq l[u]$}{$\processcosts(u,E^{\mathrm{arr}}_{\mathrm{dep}}[e]);$}
    \lIf{$u=s$}{$C[e]:=1; \Sigma[e]:=1;$}
    \If{$C[e]\neq\infty$}{
            $a:=l[v];$ \lWhile{$a\leq|E^{\mathrm{dep}}_{\mathrm{node}}[v]|\wedge\dep(E^{\mathrm{arr}}[E^{\mathrm{dep}}_{\mathrm{node}}[v][a]]) < \tau+\lambda$}{$a:=a+1;$}
            $b:=r[v];$ \lWhile{$b<|E^{\mathrm{dep}}_{\mathrm{node}}[v]|\wedge\dep(E^{\mathrm{arr}}[E^{\mathrm{dep}}_{\mathrm{node}}[v][b+1]]) \leq \tau+\lambda+\beta$}{$b:=b+1;$}
            $\processcosts(v,a-1); l_{c} := \max(a, r[v]+1);$\\
            \lWhile{$|\intervs{}[v]|>0\wedge C[e] \prec \mathrm{last}(\intervs{}[v]).c$}{$Q := \mathrm{poplast}(\intervs{}[v])$; $l_{c} := Q.l$;  \lForEach*{$f\in Q.P$}{$R[f] := a-1;$}}
            \lIf{$|\intervs{}[v]|>0\wedge \mathrm{last}(\intervs{}[v]).c = C[e]$}{$Q := \mathrm{last
            }(\intervs{}[v])$; $Q.r := b$; $Q.\eta := Q.\eta+\Sigma[e]$; $Q.P := Q.P\cup\{e\}$; $L[e] := Q.l;$}
            \lElse{\lIf*{$l_{c}\leq b$}{$L[e] := l_{c}$; $\mathrm{pushlast
            }(\intervs{}[v], (l_{c}, b, C[e], \{e\},\Sigma[e]))$;}}\label{alg3:forwardend}
            $R[e]:=b$; $r[v]:=b$;
    }
}
\lForEach{$v\in V$}{$c^{*}[v]=[\infty,\infty];\sigma^{*}[v]=0;\delta[v]:=0;$}
$c^{*}[s]:=[0,0];\sigma^{*}[s]:=1;$ \lForEach{$e\in E$}{$\Sigma^{*}[e]:=0;b[e]:=0;$}
\lForEach{$e = (u,v,\tau,\lambda)\in E^{\mathrm{arr}}$\label{alg3:pre-backward1}}{\lIf*{$(C[e]<\infty\wedge(\arr(e),C[e])\vartriangleleft c^{*}[v])$}{$c^{*}[v]:=(\arr(e),C[e]);$}}
\lForEach{$e = (u,v,\tau,\lambda)\in E^{\mathrm{arr}}$\label{alg3:pre-backward2}}{\lIf*{$(\arr(e),C[e])=c^{*}[v]$}{$\Sigma^{*}[e]:=\Sigma[e]$; $\sigma^{*}[v]:=\sigma^{*}[v]+\Sigma[e];$}}
\ForEach{$e = (u,v,\tau,\lambda)\in \mathrm{reverse}(E^{\mathrm{arr}})$}{\label{alg3:backwardstart}
    \If{$0<L[e]\leq R[e]$}{
        \lForEach{$f\in E_{\mathrm{node}}^{\mathrm{dep}}[v][\max(l[v],R[e]+1):r[v]]$}{\lIf*{$v\neq s\vee C[f] = C[e]+1$}{$\delta[v] := \delta[v] - b[f]/\Sigma[f];$}}
        $r[v]:=R[e];$ \lForEach{$f\in E_{\mathrm{node}}^{\mathrm{dep}}[v][L[e]:\min(R[e],l[v]-1)]$}{\lIf*{$v\neq s\vee C[f]=C[e]+1$}{$\delta[v] := \delta[v] + b[f]/\Sigma[f];$}}
        $l[v]:=L[e];b[e]:=\Sigma[e]\delta[v];$ 
    }
    \lIf*{$\Sigma^{*}[e]>0$}{$b[e]:=b[e]+\Sigma^{*}[e]/\sigma^{*}[v];$}
\label{alg3:backwardend}}
\Return $b$ 
\BlankLine
\SetArgSty{textbf}
\Finalize{$v,j$}:\\
\Indp
    \While{$|\intervs{}[v]|>0\wedge\mathrm{first}(\intervs{}[v]).l\leq j$}{
        $Q := \mathrm{first}(\intervs{}[v]); ql:=Q.l; q\eta:=Q.\eta;$\\
        \While{$|Q.P|>0\wedge R[\mathrm{first}(Q.P)]\leq j$}{
            $e:=\mathrm{popfirst}(Q.P)$; \lForEach{$f\in E_{\mathrm{node}}^{\mathrm{dep}}[v][ql:R[e]]$}{$C[f]:=Q.c+1;\Sigma[f]:=q\eta;$}
            $q\eta:=q\eta-\Sigma[e];ql:=R[e]+1;$ \lForEach{$f\in E_{\mathrm{node}}^{\mathrm{dep}}[v][ql:\min(j,Q.r)]$}{$C[f]:=Q.c+1;\Sigma[f]:=Q.\eta;$}
        }
        \leIf{$j\geq Q.r$}{$\mathrm{popfirst}(\intervs{}[v]);$}{$Q.l:=j+1;Q.\eta=q\eta;\mathbf{break};$}
    }
    $l[v]:=j+1;$
\caption{compute \sfo\ $b_{s,e}$, for all $e\in E$}
\label{alg:sfo_restless}
\end{algorithm*}

\section{The general algorithm}
\label{sec:generalalg}

In this section we show how Algorithm~\ref{alg:sfo_restless} can be generalised in order to deal with \fa, \fo, \sh\ and \sfa\ walks. To this aim, we first introduce the notion of cost and target cost structures and we extend Facts~\ref{fact:sfo_prefixoptimality} and~\ref{fact:sfo_toptimalimpliescoptimal} to these structures.

\subsection{Walk cost and target cost structures}
\label{sec:walkstructures}

In order to generalise Algorithm~\ref{alg:sfo_restless} to other centrality measures, we integrate a temporal graph $G=(V,E,\beta)$ with an algebraic \textit{cost structure} $(\mathcal{C},\gamma,\oplus,\preceq)$, where $\mathcal{C}$ is the set of possible \textit{cost values}, $\gamma$ is a \textit{cost function} $\gamma:E\rightarrow\mathcal{C}$, $\oplus$ is a \textit{cost combination function} $\oplus:\mathcal{C}\times\mathcal{C}\rightarrow\mathcal{C}$, and $\preceq$ is a \textit{cost total order} with $\preceq\ \subseteq\ \mathcal{C}\times\mathcal{C}$. For any two elements $c_{1}$ and $c_{2}$ of $\mathcal{C}$, we say that $c_{1}=c_{2}$ if $c_{1}\preceq c_{2}$ and $c_{2}\preceq c_{1}$ both hold. We also define the relation $\prec$ between the elements of $\mathcal{C}$ as $c_{1}\prec c_{2}$ if and only if $c_{1}\preceq c_{2}$ and $c_{1}\neq c_{2}$. For any walk $W=\langle e_1,\ldots,e_k\rangle$, the \textit{cost function} of $W$ is recursively defined as follows: $\gamma(W) = \gamma(\langle e_1,\ldots,e_{k-1}\rangle)\oplus\gamma(e_k)$, with $\gamma(\langle e_1\rangle)=\gamma(e_1)$ (in other words, the costs combine along the walk according to the cost combination function). The cost structure is supposed to satisfy the following \textit{strict right-isotonicity property}~\cite{BrunelliCV2021,Sobrinho2005,Griff2010} (\textit{isotonicity} for short): for any $c_1,c_2,c\in \mathcal{C}$ such that $c_1\prec c_2$, we have  $c_1\oplus c\prec c_2\oplus c$. This property implies the following \textit{walk extension} property: for any two walks $W$ and $X$ such that $\gamma(W)\prec \gamma(X)$ and for any temporal edge $e$ which can extend both $W$ and $X$, we have $\gamma(W.e)\prec\gamma(X.e)$ (that is, if several walks are extended by a given temporal edge $e$, then the best cost is obtained only by extending a walk with minimum cost). The isotonicity property also implies the following \textit{prefix} property, which generalizes Fact~\ref{fact:sfo_prefixoptimality} and is similar to the prefix-optimality property introduced in \cite{Rymar2023}.

\begin{fact}\label{fact:prefixoptimality}
Let $G=(V,E,\beta)$ be a temporal graph and $(\mathcal{C},\gamma,\oplus,\preceq)$ be a cost structure satisfying the isotonicity property. For any node $s\in V$, if a walk $W$ with last temporal edge $f\in E$ has minimum cost among the $sf$-walks, and $e\in E$ is a temporal edge of $W$, then the prefix of $W$ up to the temporal edge $e$ has minimum cost among the $se$-walks.
\end{fact}

\begin{proof}
Let $W_{1}$ be the prefix of $W$ up to the temporal edge $e$ and let us prove that $W_{1}$ has minimum cost among the $se$-walks. To this aim, let $\langle e_{1}e_{2}\cdots e_{k}\rangle$, with $e_{k}=f$ be the suffix $W_{2}$ of $W$ following the temporal edge $e$. Suppose that $W_{1}$ has not minimum cost and that there exists another $se$-walk such that $\gamma(X)\prec\gamma(W_1)$. By the isotonicity of $\cost$ it follows that $\gamma(X.e_{1})\prec\gamma(W_1.e_{1})$. By continuing in this way, we have that $\gamma(X.e_{1}.e_{2}.\cdots.e_{k})\prec\gamma(W_1.e_{1}.e_{2}.\cdots.e_{k})=\gamma(W)$, contradicting the fact that $W$ has minimum cost among the $sf$-walks. The fact thus follows.
\end{proof}

\begin{table*}
\centering
\begin{adjustbox}{width=\textwidth}
\begin{tabular}{||l||c|c|c|c||}
\cline{2-5}
\multicolumn{1}{c||}{} & $\mathcal{C}$ & $\gamma(e)$ & $c_{1}\oplus c_{2}$ & $c_{1}\preceq c_{2}$ \\
\hline
\hline
All & $\{0\}$ & $0$ & $0$ & $\mathtt{true}$ \\
\hline
Shortest & $\mathbf{N}$ & $1$ & $c_{1}+c_{2}$ & $c_{1}\leq c_{2}$ \\
\hline
Latest & $\mathbf{Z}$ & $-\mathrm{dep}(e)$ & $c_{1}$ & $c_{1}\leq c_{2}$ \\
\hline
Shortest latest & $\mathbf{Z}\times\mathbf{N}$ & $(-\mathrm{dep}(e), 1)$ & $(c_{1}[1], c_{1}[2]+c_{2}[2])$ & $c_{1}[1]<c_{2}[1]\ \mathbf{or}\ (c_{1}[1]= c_{2}[1]\ \mathbf{and}\ c_{1}[2]\leq c_{2}[2])$ \\
\hline
\end{tabular}
\end{adjustbox}
\caption{The 4 cost structures used in this paper}
\label{tab:coststructures}
\end{table*}

\begin{table*}
\centering
\begin{tabular}{||l||c|c||}
\cline{2-3}
\multicolumn{1}{c||}{} & $\mathcal{C}^{F}$ & $c_{1}\preceq^{F}c_{2}$ \\
\hline
\hline
Natural & $\mathbf{Z}$ & $c_{1}\leq c_{2}$   \\
\hline
Lexicographic & $\mathbf{Z}\times\mathbf{N}$ & $c_{1}[1]< c_{2}[1]\ \mathbf{or}\ (c_{1}[1]= c_{2}[1]\ \mathbf{and}\ c_{1}[2]\leq c_{2}[2])$ \\
\hline
\end{tabular}
\caption{The 2 target cost structures along with the $\tcf$ functions used in this paper}
\label{tab:targetcoststructures}
\end{table*}

In this paper, we will consider the four cost structures shown in Table~\ref{tab:coststructures}. It is easy to verify that all of them satisfy the isotonicity property. For example, let us prove that the shortest latest cost structure satisfies the isotonicity property. Suppose that $c_{1}, c_{2}\in\mathbf{Z}\times\mathbf{N}$ satisfy $c_{1}\prec c_{2}$: this implies that either $c_{1}[1]>c_{2}[1]$ or $c_{1}[1]=c_{2}[1]\wedge c_{1}[2]<c_{2}[2]$. For any $c\in\mathbf{Z}\times\mathbf{N}$, $c_{1}\oplus c=(c_{1}[1], c_{1}[2]+c[2])$ and $c_{2}\oplus c=(c_{2}[1], c_{2}[2]+c[2])$. If $c_{1}[1]>c_{2}[1]$, then $c_{1}\oplus c\prec c_{2}\oplus c$. Otherwise (that is, $c_{1}[1]=c_{2}[1]$ and $c_{1}[2]<c_{2}[2]$), $c_{1}[2]+c[2]<c_{2}[2]+c[2]$ and, thus, $c_{1}\oplus c\prec c_{2}\oplus c$. 

Similar to the algorithms introduced in Section~\ref{sec:algorithms}, the general algorithm consists of three phases. Given a temporal graph $G=(V,E,\beta)$ and a cost structure $(\mathcal{C},\gamma,\oplus,\preceq)$ among the ones in Table~\ref{tab:coststructures}, in the first  phase (that is, the forward phase), the algorithm counts, for any source node $s$ and for any temporal edge $e$, the number of $se$-walks which are optimal with respect to the walk cost function $\gamma$. In the other two phases (that is, the intermediate and the backward phase), the algorithm makes use of a \textit{target} (or \textit{final}) \textit{cost structure} $(\mathcal{C}^{F},\preceq^{F})$, where $\mathcal{C}^{F}$ is the set of possible \textit{target cost values} and $\preceq^{F}$ is a \textit{target cost total order} with $\preceq^{F}\ \subseteq\ \mathcal{C}^{F}\times \mathcal{C}^{F}$. For any two elements $c_{1}$ and $c_{2}$ of $\mathcal{C}^{F}$, we say that $c_{1}=^{F}c_{2}$ if $c_{1}\preceq^{F}c_{2}$ and $c_{2}\preceq^{F}c_{1}$ both hold. We also define the relation $\prec^{F}$ between the elements of $\mathcal{C}^{F}$ as $c_{1}\prec^{F}c_{2}$ if and only if $c_{1}\preceq^{F}c_{2}$ and $c_{1}\neq^{F}c_{2}$. 

Costs and target costs are related through a function $\tcf: E\times \mathcal{C} \rightarrow \mathcal{C}^{F}$ that associates to a temporal edge and a cost, a corresponding target cost. This function needs to satisfy the following \textit{increasing} property: for any $c_1,c_2\in \mathcal{C}$ such that $c_1\prec c_2$ and for any $e\in E$, we have that $\tcf(e,c_1) \prec^{F} \tcf(e,c_2)$. In this paper, we will consider the two target cost structures shown in Table~\ref{tab:targetcoststructures}, which together with a combination of one of the cost structures of Table ~\ref{tab:coststructures} and an appropriate $\tcf$ function will allow us to deal all optimality criteria. For example, \fa\ walks can be modeled by using the latest cost structure of Table~\ref{tab:coststructures} that associates to a walk the opposite of its departure time as a cost with later time being considered as lower cost: for that, it suffices to define the cost of a temporal edge as the opposite of its departure time and a combination function $\oplus$ that returns its first argument. The duration of a walk $W$ with cost $c$ and ending with a temporal edge $e=(u,v,\tau,\lambda)$ is then obtained by using the natural target cost (see Table~\ref{tab:targetcoststructures}) and the target cost function $\tcf(e,c)=\arr(e)+c$, since $\arr(e)=\tau+\lambda$ equals $\arr(W)$ and $c=-\dep(W)$. As shown in Table~\ref{tab:walkmeasures}, any of the walk optimality criteria defined in Section~\ref{sec:preliminaries} can be modelled by an appropriate combination of the introduced cost structures, target cost structures, and target functions.

Given a cost structure $\cost=(\mathcal{C},\gamma,\oplus,\preceq)$ and a target cost structure $\tcost=(\mathcal{C}^{F},\preceq^{F})$ along with a target cost function $\tcf$, we say that a temporal $st$-walk $W$ is \textit{$\tcost$-optimal} if, for any $st$-walk $X$, we have $\tcf(e,\gamma(W))\preceq^{F}\tcf(f,\gamma(X))$, where $e$ and $f$ are the last temporal edge of $W$ and $X$, respectively. We also say that a temporal $se$-walk $W$ is $\cost$-optimal if, for any $se$-walk $X$, we have $\gamma(W)\preceq\gamma(X)$. Note that a $\cost$-optimal $se$-walk $X$ not necessarily is a $\tcost$-optimal $st$-walk where $t$ is the head of $e$. The opposite, instead, is true as shown by the following which is a generalization of Fact~\ref{fact:sfo_toptimalimpliescoptimal}.

\begin{fact}\label{fact:toptimalimpliescoptimal}
Given a temporal graph ${G}=(V,{E},\beta)$, a cost structure $(\mathcal{C},\gamma,\oplus,\preceq)$, a target cost structure $(\mathcal{C}^{F},\preceq^{F})$, and a $\tcf$ function, let $W$ be a $\tcost$-optimal $st$-walk (for some $s,t\in V$) and let $e$ be the last temporal edge of $W$. Then, $W$ is a $\cost$-optimal $se$-walk.
\end{fact}

\begin{proof}
Suppose that $W$ is not $\cost$-optimal and that there exists another $se$-walk $X$ (and, hence, $st$-walk) such that $\gamma(X)\prec\gamma(W)$. By the increasing property of the $\tcf$ function it follows that $\tcf(e,\gamma(X))\prec^{F}\tcf(e,\gamma(W))$, contradicting the fact that $W$ is a $\tcost$-optimal $st$-walk. The fact thus follows.
\end{proof}

\begin{table}[ht]
\centering
    \begin{tabular}{l|l|l|c}
Optimality criterion & Cost & Target cost & $\tcf(e, c)$ \\
\hline
Shortest & Shortest & Natural & $c$ \\
Foremost & All & Natural & $\mathrm{arr}(e)$ \\
Latest & Latest & Natural & $c$ \\
Fastest & Latest & Natural & $\mathrm{arr}(e)+c$ \\
Shortest foremost & Shortest & Lexicographic & $(\mathrm{arr}(e), c)$\\
Shortest latest & Shortest latest & Lexicographic & $c$\\
Shortest fastest & Shortest latest & Lexicographic & $(\mathrm{arr}(e)+c[1], c[2])$\\
\end{tabular}

\caption{Optimality criteria in terms of cost and target cost structures. In the last column of the table, the semantic of $c$ is the following: number of hops for \sh\ and for \sfo, the opposite of the departure time for latest and for \fa\, the array whose first entry is the opposite of the departure time and whose second entry is the number of hops for shortest latest and \sfa.}
\label{tab:walkmeasures}
\end{table}

Moreover, all the definitions concerning the \sh\ and the \sfo\ betweenness can be appropriately adapted to the any other optimality criterion case and all the results proved for the \sfo\ betweenness can be proved also for the corresponding betweenness.

\subsection{The general algorithm}

The pseudo-code of the general algorithm for any of optimality criteria which fits in our framework is shown in Algorithm~\ref{alg:restless}. Note that the pseudo-code uses the symbol $\infty^{\cost}$ and $\infty^{\tcost}$ to denote the natural infinite value of a cost structure $\cost$ and a target cost structure $\tcost$. The pseudo-code also uses the symbol $0^{\tcost}$ to denote the natural minimum value of a target cost structure $\tcost$.  Moreover, note the strong similarity of the pseudo-code with Algorithm~\ref{alg:sfo_restless}: indeed, the two pseudo-codes are almost the same apart from the use in Algorithm~\ref{alg:restless} of the components of the cost and target cost structures. Even in this case, we do not include the control structure for dealing with numerical approximation problems, which is not needed if big integer and big rational data structures are used.

\begin{algorithm*}[ht]
\small
\Input{temporal graph $G=(V,E,\beta)$ (represented by $E^{\mathrm{dep}}$ and $E^{\mathrm{arr}}$), $s\in V$, cost structure $\cost=(\mathcal{C},\gamma,\oplus,\preceq)$, target cost structure $\tcost=(\mathcal{C}^{F},\preceq^{F})$, and target cost function $\tcf$}
\Output{$s$-temporal betweenness $b_{s,e}$ of each $e\in E$ w.r.t. \cost, \tcost, and $\tcf$}

Compute the lists $E^{\mathrm{dep}}_{\mathrm{node}}$ and $E^{\mathrm{arr}}_{\mathrm{dep}};$\\
\lForEach{$v\in V$}{$l[v]:=1;r[v]:=0;\intervs{}[v]:=\emptyset;$}
\lForEach{$e\in E$}{$L[e]:=\left|E^{arr}\right|+1;R[e]:=0;C[e]:=\infty^{\cost};\Sigma[e]:=0;$}
\ForEach{$e = (u,v,\tau,\lambda)\in E^{\mathrm{arr}}$}{\label{alg4:forwardstart}
    \lIf{$E^{\mathrm{arr}}_{\mathrm{dep}}[e]\geq l[u]$}{$\processcosts(u,E^{\mathrm{arr}}_{\mathrm{dep}}[e]);$}
    \lIf{$u=s$}{\leIf*{$\gamma(e)\prec C[e]$}{$C[e]:=\gamma(e); \Sigma[e]:=1;$}{\lIf*{$\gamma(e)=C[e]$}{$\Sigma[e]:=\Sigma[e]+1;$}}}
    \If{$C[e]\neq\infty^{\cost}$}{
            $a:=l[v];$ \lWhile{$a\leq|E^{\mathrm{dep}}_{\mathrm{node}}[v]|\wedge\dep(E^{\mathrm{arr}}[E^{\mathrm{dep}}_{\mathrm{node}}[v][a]]) < \tau+\lambda$}{$a:=a+1;$}
            $b:=r[v];$ \lWhile{$b<|E^{\mathrm{dep}}_{\mathrm{node}}[v]|\wedge\dep(E^{\mathrm{arr}}[E^{\mathrm{dep}}_{\mathrm{node}}[v][b+1]]) \leq \tau+\lambda+\beta$}{$b:=b+1;$}
            $\processcosts(v,a-1); l_{c} := \max(a, r[v]+1);$\\
            \lWhile{$|\intervs{}[v]|>0\wedge C[e] \prec \mathrm{last}(\intervs{}[v]).c$}{$Q := \mathrm{poplast}(\intervs{}[v])$; $l_{c} := Q.l$;  \lForEach*{$f\in Q.P$}{$R[f] := a-1;$}}
            \lIf{$|\intervs{}[v]|>0\wedge \mathrm{last}(\intervs{}[v]).c = C[e]$}{$Q := \mathrm{last
            }(\intervs{}[v])$; $Q.r := b$; $Q.\eta := Q.\eta+\Sigma[e]$; $Q.P := Q.P\cup\{e\}$; $L[e] := Q.l;$}
            \lElse{\lIf*{$l_{c}\leq b$}{$L[e] := l_{c}$; $\mathrm{pushlast
            }(\intervs{}[v], (l_{c}, b, C[e], \{e\},\Sigma[e]))$;}}\label{alg4:forwardend}
            $R[e]:=b$; $r[v]:=b$;
    }
}
\lForEach{$v\in V$}{$c^{*}[v]=\infty^{\Theta};\sigma^{*}[v]=0;\delta[v]:=0;$}
$c^{*}[s]:=0^{\Theta};\sigma^{*}[s]:=1;$ \lForEach{$e\in E$}{$\Sigma^{*}[e]:=0;b[e]:=0;$}
\lForEach{$e = (u,v,\tau,\lambda)\in E^{\mathrm{arr}}$\label{alg4:pre-backward1}}{\lIf*{$(C[e]<\infty\wedge\tcf(e,C[e])\prec^{F}c^{*}[v])$}{$c^{*}[v]:=\tcf(e,C[e]);$}}
\lForEach{$e = (u,v,\tau,\lambda)\in E^{\mathrm{arr}}$\label{alg4:pre-backward2}}{\lIf*{$\tcf(e,C[e])=^{F}c^{*}[v]$}{$\Sigma^{*}[e]:=\Sigma[e]$; $\sigma^{*}[v]:=\sigma^{*}[v]+\Sigma[e];$}}
\ForEach{$e = (u,v,\tau,\lambda)\in \mathrm{reverse}(E^{\mathrm{arr}})$}{\label{alg4:backwardstart}
    \If{$0<L[e]\leq R[e]$}{
        \lForEach{$f\in E_{\mathrm{node}}^{\mathrm{dep}}[v][\max(l[v],R[e]+1):r[v]]$}{\lIf*{$v\neq s\vee C[f] =
        C[e]\oplus\gamma(f)$}{$\delta[v] := \delta[v] - b[f]/\Sigma[f];$}}
        $r[v]:=R[e];$ \lForEach{$f\in E_{\mathrm{node}}^{\mathrm{dep}}[v][L[e]:\min(R[e],l[v]-1)]$}{\lIf*{$v\neq s\vee C[f]=C[e]\oplus\gamma(f)$}{$\delta[v] := \delta[v] + b[f]/\Sigma[f];$}}
        $l[v]:=L[e];b[e]:=\Sigma[e]\delta[v];$ 
    }
    \lIf*{$\Sigma^{*}[e]>0$}{$b[e]:=b[e]+\Sigma^{*}[e]/\sigma^{*}[v];$}
\label{alg4:backwardend}}
\Return $b$ 
\BlankLine
\SetArgSty{textbf}
\Finalize{$v,j$}:\\
\Indp
    \While{$|\intervs{}[v]|>0\wedge\mathrm{first}(\intervs{}[v]).l\leq j$}{
        $Q := \mathrm{first}(\intervs{}[v]);ql:=Q.l; q\eta:=Q.\eta;$\\
        \While{$|Q.P|>0\wedge R[\mathrm{first}(Q.P)]\leq j$}{
            $e:=\mathrm{popfirst}(Q.P)$; \lForEach*{$f\in E_{\mathrm{node}}^{\mathrm{dep}}[v][ql:R[e]]$}{$C[f]:=Q.c\oplus\gamma(f);\Sigma[f]:=q\eta;$}\\
            $q\eta:=Q.\eta-\Sigma[e];ql:=R[e]+1;$ \lForEach*{$f\in E_{\mathrm{node}}^{\mathrm{dep}}[v][ql:\min(j,Q.r)]$}{$C[f]:=Q.c\oplus\gamma(f);\Sigma[f]:=q\eta;$}
        }
        \leIf{$j\geq Q.r$}{$\mathrm{popfirst}(\intervs{}[v]);$}{$Q.l:=j+1;Q.\eta:=q\eta;
        \mathbf{break};$}
    }
    $l[v]:=j+1;$
\caption{compute $b_{s,e}$ of all temporal edges}
\label{alg:restless}
\end{algorithm*}

\section{Comparing \sfo\ BMNR and Fast times}
\label{sec:experiment1sfb}

\begin{center}
\begin{tabular}{@{}lrrr@{}}\toprule
\textbf{Temporal graph} & \multicolumn{1}{c}{$\mathbf{t}_{\mathbf{\textsc{BMNR}}}$} & \multicolumn{1}{c}{$\mathbf{t}_{\mathbf{\textsc{Fast}}}$} & \multicolumn{1}{c}{$\mathbf{t}_{\mathbf{\textsc{BMNR}}}/\mathbf{t}_{\mathbf{\textsc{Fast}}}$} \\\midrule
\texttt{Infectious} & 3158.12 & 1595.56 & 1.98 \\
\texttt{Digg reply} & 1201.86 & 507.74 & 2.37 \\
\texttt{Facebook wall} & 3402.02 & 1354.76 & 2.51 \\
\texttt{Slashdot reply} & 4518.64 & 1730.93 & 2.61 \\
\texttt{SMS} & 12501.27 & 4773.57 & 2.62 \\
\texttt{Wiki elections} & 526.94 & 125.10 & 4.21 \\
\texttt{College msg} & 246.30 & 27.95 & 8.81 \\
\texttt{Topology} & 8926.13 & 930.49 & 9.59 \\
\texttt{Hypertext 2009} & 85.17 & 1.09 & 78.42 \\
\texttt{High school 2011} & 134.05 & 1.63 & 82.34 \\
\texttt{High school 2012} & 402.18 & 3.71 & 108.54 \\
\texttt{Primary school} & 1938.72 & 17.29 & 112.16 \\
\texttt{Email EU} & 12233.57 & 91.78 & 133.30 \\
\texttt{Hospital ward} & 204.56 & 1.14 & 178.88 \\
\texttt{High school 2013} & 7329.39 & 38.94 & 188.21 \\
\bottomrule
\end{tabular}

\end{center}

\section{Ranking correlations}
\label{sec:rankingcorrelation}

In the Tables~\ref{tbl:tfab}-\ref{tbl:9999}, we show the different correlations (that is, Kendall, weighted Kendall, and intersection correlation) between the rankings produced by the execution of our algorithm (that is, Algorithm~\ref{alg:restless}) on the networks of the first dataset, for different optimality criteria (that is, fastest, foremost, shortest, shortest fastest, and shortest foremost) and for different values of $\beta$ (that is, $\beta=300,600,1200,2400,\infty$). In the Tables~\ref{tbl:300pt}-\ref{tbl:9999pt}, we show the same correlations  between the rankings produced by the execution of the algorithm on the public transport networks, for different values of $\beta$ and for different pairs of optimality criteria.

\begin{table*}
\begin{tabular}{||l||r|r|r|r|r|r||}
\cline{2-7}
\multicolumn{1}{c||}{} & \multicolumn{6}{|c||}{Kendall's $\tau$}\\
\cline{2-7}
\multicolumn{1}{c||}{} & \multicolumn{3}{|c|}{$300$} & \multicolumn{2}{|c|}{$600$} & \multicolumn{1}{|c||}{$1200$}\\
\hline
Network & $600$ & $1200$ & $2400$ & $1200$ & $2400$ & $2400$\\
\hline\hline
\texttt{Hypertext 2009} & 0.87 & 0.86 & 0.80 & 0.88 & 0.81 & 0.89 \\
\texttt{Hospital ward} & 0.88 & 0.88 & 0.84 & 0.97 & 0.89 & 0.90 \\
\texttt{Highschool 2011} & 0.87 & 0.72 & 0.71 & 0.75 & 0.72 & 0.87 \\
\texttt{Highschool 2012} & 0.77 & 0.69 & 0.70 & 0.80 & 0.74 & 0.80 \\
\texttt{Primary school} & 0.92 & 0.89 & 0.89 & 0.96 & 0.96 & 1.00 \\
\hline\hline
\end{tabular}
\caption{Fastest walks and Kendall, weighted Kendall, and intersection correlation}
\label{tbl:tfab}
\end{table*}

\begin{table*}
\begin{tabular}{||l||r|r|r|r|r|r||}
\cline{2-7}
\multicolumn{1}{c||}{} & \multicolumn{6}{|c||}{Kendall's $\tau$}\\
\cline{2-7}
\multicolumn{1}{c||}{} & \multicolumn{3}{|c|}{$300$} & \multicolumn{2}{|c|}{$600$} & \multicolumn{1}{|c||}{$1200$}\\
\hline
Network & $600$ & $1200$ & $2400$ & $1200$ & $2400$ & $2400$\\
\hline\hline
\texttt{Hypertext 2009} & 0.77 & 0.74 & 0.72 & 0.84 & 0.82 & 0.91 \\
\texttt{Hospital ward} & 0.92 & 0.87 & 0.82 & 0.93 & 0.87 & 0.91 \\
\texttt{Highschool 2011} & 0.85 & 0.65 & 0.64 & 0.68 & 0.66 & 0.92 \\
\texttt{Highschool 2012} & 0.78 & 0.71 & 0.71 & 0.80 & 0.72 & 0.77 \\
\texttt{Primary school} & 0.86 & 0.86 & 0.86 & 0.92 & 0.92 & 1.00 \\
\hline\hline
\end{tabular}
\caption{Foremost walks and Kendall, weighted Kendall, and intersection correlation}
\end{table*}

\begin{table*}
\begin{tabular}{||l||r|r|r|r|r|r|r|r|r|r||}
\cline{2-11}
\multicolumn{1}{c||}{} & \multicolumn{10}{|c||}{Kendall's $\tau$}\\
\cline{2-11}
\multicolumn{1}{c||}{} & \multicolumn{4}{c|}{$\infty$} & \multicolumn{3}{|c|}{$300$} & \multicolumn{2}{|c|}{$600$} & \multicolumn{1}{|c||}{$1200$}\\
\hline
Network & $300$ & $600$ & $1200$ & $2400$ & $600$ & $1200$ & $2400$ & $1200$ & $2400$ & $2400$\\
\hline\hline
\texttt{Hypertext 2009} & 0.68 & 0.71 & 0.75 & 0.75 & 0.84 & 0.78 & 0.74 & 0.82 & 0.78 & 0.88 \\
\texttt{Hospital ward} & 0.75 & 0.71 & 0.70 & 0.69 & 0.76 & 0.72 & 0.69 & 0.82 & 0.76 & 0.80 \\
\texttt{Highschool 2011} & 0.65 & 0.67 & 0.70 & 0.74 & 0.83 & 0.76 & 0.76 & 0.78 & 0.78 & 0.87 \\
\texttt{Highschool 2012} & 0.57 & 0.56 & 0.53 & 0.57 & 0.76 & 0.67 & 0.70 & 0.75 & 0.71 & 0.72 \\
\texttt{Primary school} & 0.68 & 0.70 & 0.72 & 0.77 & 0.85 & 0.81 & 0.78 & 0.88 & 0.84 & 0.90 \\
\hline\hline
\end{tabular}
\caption{Shortest walks and Kendall, weighted Kendall, and intersection correlation}
\end{table*}

\begin{table*}
\begin{tabular}{||l||r|r|r|r|r|r|r|r|r|r||}
\cline{2-11}
\multicolumn{1}{c||}{} & \multicolumn{10}{|c||}{Kendall's $\tau$}\\
\cline{2-11}
\multicolumn{1}{c||}{} & \multicolumn{4}{c|}{$\infty$} & \multicolumn{3}{|c|}{$300$} & \multicolumn{2}{|c|}{$600$} & \multicolumn{1}{|c||}{$1200$}\\
\hline
Network & $300$ & $600$ & $1200$ & $2400$ & $600$ & $1200$ & $2400$ & $1200$ & $2400$ & $2400$\\
\hline\hline
\texttt{Hypertext 2009} & 0.83 & 0.86 & 0.92 & 0.96 & 0.87 & 0.85 & 0.82 & 0.89 & 0.86 & 0.93 \\
\texttt{Hospital ward} & 0.66 & 0.70 & 0.78 & 0.82 & 0.79 & 0.77 & 0.71 & 0.84 & 0.77 & 0.83 \\
\texttt{Highschool 2011} & 0.81 & 0.82 & 0.83 & 0.88 & 0.86 & 0.80 & 0.78 & 0.83 & 0.80 & 0.85 \\
\texttt{Highschool 2012} & 0.71 & 0.69 & 0.66 & 0.76 & 0.78 & 0.69 & 0.72 & 0.76 & 0.72 & 0.74 \\
\texttt{Primary school} & 0.88 & 0.96 & 0.99 & 0.99 & 0.90 & 0.89 & 0.88 & 0.97 & 0.96 & 0.99 \\
\hline\hline
\end{tabular}
\caption{Shortest fastest walks and Kendall, weighted Kendall, and intersection correlation}
\end{table*}

\begin{table*}
\begin{tabular}{||l||r|r|r|r|r|r|r|r|r|r||}
\cline{2-11}
\multicolumn{1}{c||}{} & \multicolumn{10}{|c||}{Kendall's $\tau$}\\
\cline{2-11}
\multicolumn{1}{c||}{} & \multicolumn{4}{c|}{$\infty$} & \multicolumn{3}{|c|}{$300$} & \multicolumn{2}{|c|}{$600$} & \multicolumn{1}{|c||}{$1200$}\\
\hline
Network & $300$ & $600$ & $1200$ & $2400$ & $600$ & $1200$ & $2400$ & $1200$ & $2400$ & $2400$\\
\hline\hline
\texttt{Hypertext 2009} & 0.59 & 0.64 & 0.71 & 0.81 & 0.70 & 0.69 & 0.61 & 0.74 & 0.67 & 0.80 \\
\texttt{Hospital ward} & 0.64 & 0.69 & 0.74 & 0.76 & 0.77 & 0.70 & 0.65 & 0.81 & 0.73 & 0.79 \\
\texttt{Highschool 2011} & 0.57 & 0.57 & 0.63 & 0.71 & 0.80 & 0.72 & 0.65 & 0.73 & 0.64 & 0.77 \\
\texttt{Highschool 2012} & 0.56 & 0.54 & 0.53 & 0.63 & 0.75 & 0.66 & 0.65 & 0.74 & 0.64 & 0.66 \\
\texttt{Primary school} & 0.48 & 0.65 & 0.80 & 0.87 & 0.62 & 0.51 & 0.49 & 0.74 & 0.69 & 0.86 \\
\hline\hline
\end{tabular}
\caption{Shortest foremost walks and Kendall, weighted Kendall, and intersection correlation}
\end{table*}

\begin{table*}
\begin{tabular}{||l||r|r|r|r|r|r|r|r|r|r||}
\cline{2-11}
\multicolumn{1}{c||}{} & \multicolumn{10}{|c||}{Kendall's $\tau$}\\
\cline{2-11}
\multicolumn{1}{c||}{} & \multicolumn{4}{c|}{Fa} & \multicolumn{3}{|c|}{Fo} & \multicolumn{2}{|c|}{Sh} & \multicolumn{1}{|c||}{SFa}\\
\hline
Network & Fo & Sh & SFa & SFo & Sh & SFa & SFo & SFa & SFo & SFo\\
\hline\hline
\texttt{Hypertext 2009} & 0.64 & 0.84 & 0.84 & 0.66 & 0.58 & 0.57 & 0.67 & 0.93 & 0.69 & 0.68 \\
\texttt{Hospital ward} & 0.79 & 0.79 & 0.79 & 0.73 & 0.71 & 0.70 & 0.76 & 0.94 & 0.82 & 0.80 \\
\texttt{Highschool 2011} & 0.70 & 0.79 & 0.81 & 0.72 & 0.62 & 0.63 & 0.70 & 0.92 & 0.80 & 0.78 \\
\texttt{Highschool 2012} & 0.79 & 0.82 & 0.82 & 0.78 & 0.71 & 0.71 & 0.74 & 0.94 & 0.87 & 0.87 \\
\texttt{Primary school} & 0.44 & 0.69 & 0.71 & 0.50 & 0.32 & 0.31 & 0.48 & 0.82 & 0.48 & 0.47 \\
\hline\hline
\end{tabular}
\caption{$\beta=300$ and Kendall, weighted Kendall, and intersection correlation}
\end{table*}

\begin{table*}
\begin{tabular}{||l||r|r|r|r|r|r|r|r|r|r||}
\cline{2-11}
\multicolumn{1}{c||}{} & \multicolumn{10}{|c||}{Kendall's $\tau$}\\
\cline{2-11}
\multicolumn{1}{c||}{} & \multicolumn{4}{c|}{Fa} & \multicolumn{3}{|c|}{Fo} & \multicolumn{2}{|c|}{Sh} & \multicolumn{1}{|c||}{SFa}\\
\hline
Network & Fo & Sh & SFa & SFo & Sh & SFa & SFo & SFa & SFo & SFo\\
\hline\hline
\texttt{Hypertext 2009} & 0.60 & 0.80 & 0.78 & 0.63 & 0.52 & 0.50 & 0.62 & 0.88 & 0.68 & 0.66 \\
\texttt{Hospital ward} & 0.82 & 0.71 & 0.72 & 0.69 & 0.65 & 0.65 & 0.68 & 0.93 & 0.81 & 0.81 \\
\texttt{Highschool 2011} & 0.72 & 0.75 & 0.79 & 0.73 & 0.60 & 0.62 & 0.68 & 0.89 & 0.79 & 0.80 \\
\texttt{Highschool 2012} & 0.80 & 0.78 & 0.79 & 0.76 & 0.68 & 0.68 & 0.71 & 0.94 & 0.86 & 0.86 \\
\texttt{Primary school} & 0.40 & 0.66 & 0.70 & 0.45 & 0.29 & 0.30 & 0.35 & 0.79 & 0.48 & 0.50 \\
\hline\hline
\end{tabular}
\caption{$\beta=600$ and Kendall, weighted Kendall, and intersection correlation}
\end{table*}

\begin{table*}
\begin{tabular}{||l||r|r|r|r|r|r|r|r|r|r||}
\cline{2-11}
\multicolumn{1}{c||}{} & \multicolumn{10}{|c||}{Kendall's $\tau$}\\
\cline{2-11}
\multicolumn{1}{c||}{} & \multicolumn{4}{c|}{Fa} & \multicolumn{3}{|c|}{Fo} & \multicolumn{2}{|c|}{Sh} & \multicolumn{1}{|c||}{SFa}\\
\hline
Network & Fo & Sh & SFa & SFo & Sh & SFa & SFo & SFa & SFo & SFo\\
\hline\hline
\texttt{Hypertext 2009} & 0.57 & 0.75 & 0.77 & 0.63 & 0.48 & 0.48 & 0.61 & 0.83 & 0.64 & 0.64 \\
\texttt{Hospital ward} & 0.77 & 0.62 & 0.63 & 0.59 & 0.56 & 0.58 & 0.61 & 0.87 & 0.78 & 0.80 \\
\texttt{Highschool 2011} & 0.72 & 0.73 & 0.75 & 0.67 & 0.57 & 0.57 & 0.68 & 0.87 & 0.71 & 0.69 \\
\texttt{Highschool 2012} & 0.82 & 0.76 & 0.77 & 0.73 & 0.68 & 0.68 & 0.71 & 0.93 & 0.84 & 0.84 \\
\texttt{Primary school} & 0.38 & 0.64 & 0.70 & 0.38 & 0.27 & 0.29 & 0.29 & 0.75 & 0.41 & 0.45 \\
\hline\hline
\end{tabular}
\caption{$\beta=1200$ and Kendall, weighted Kendall, and intersection correlation}
\end{table*}

\begin{table*}
\begin{tabular}{||l||r|r|r|r|r|r|r|r|r|r||}
\cline{2-11}
\multicolumn{1}{c||}{} & \multicolumn{10}{|c||}{Kendall's $\tau$}\\
\cline{2-11}
\multicolumn{1}{c||}{} & \multicolumn{4}{c|}{Fa} & \multicolumn{3}{|c|}{Fo} & \multicolumn{2}{|c|}{Sh} & \multicolumn{1}{|c||}{SFa}\\
\hline
Network & Fo & Sh & SFa & SFo & Sh & SFa & SFo & SFa & SFo & SFo\\
\hline\hline
\texttt{Hypertext 2009} & 0.55 & 0.75 & 0.79 & 0.62 & 0.46 & 0.45 & 0.55 & 0.79 & 0.65 & 0.63 \\
\texttt{Hospital ward} & 0.78 & 0.62 & 0.59 & 0.57 & 0.59 & 0.55 & 0.59 & 0.85 & 0.79 & 0.77 \\
\texttt{Highschool 2011} & 0.68 & 0.67 & 0.69 & 0.61 & 0.54 & 0.51 & 0.60 & 0.81 & 0.66 & 0.64 \\
\texttt{Highschool 2012} & 0.78 & 0.71 & 0.69 & 0.65 & 0.62 & 0.61 & 0.66 & 0.89 & 0.76 & 0.77 \\
\texttt{Primary school} & 0.38 & 0.63 & 0.70 & 0.36 & 0.26 & 0.29 & 0.27 & 0.72 & 0.38 & 0.44 \\
\hline\hline
\end{tabular}
\caption{$\beta=2400$ and Kendall, weighted Kendall, and intersection correlation}
\end{table*}

\begin{table*}
\begin{tabular}{||l||r|r|r||}
\cline{2-4}
\multicolumn{1}{c||}{} & \multicolumn{3}{|c||}{Kendall's $\tau$}\\
\cline{2-4}
\multicolumn{1}{c||}{} & \multicolumn{2}{c|}{Sh} & \multicolumn{1}{|c||}{SFa}\\
\hline
Network & SFa & SFo & SFo \\
\hline\hline
\texttt{Hypertext 2009} & 0.72 & 0.55 & 0.56 \\
\texttt{Hospital ward} & 0.65 & 0.63 & 0.70 \\
\texttt{Highschool 2011} & 0.70 & 0.54 & 0.52 \\
\texttt{Highschool 2012} & 0.69 & 0.49 & 0.60 \\
\texttt{Primary school} & 0.62 & 0.28 & 0.41 \\
\hline\hline
\end{tabular}
\caption{$\beta=\infty$ and Kendall, weighted Kendall, and intersection correlation}
\label{tbl:9999}
\end{table*}

\begin{table*}
\begin{tabular}{||l||r|r|r|r|r|r|r|r|r|r||}
\cline{2-11}
\multicolumn{1}{c||}{} & \multicolumn{10}{|c||}{Kendall's $\tau$}\\
\cline{2-11}
\multicolumn{1}{c||}{} & \multicolumn{4}{c|}{Fa} & \multicolumn{3}{|c|}{Fo} & \multicolumn{2}{|c|}{Sh} & \multicolumn{1}{|c||}{SFa}\\
\hline
Network & Fo & Sh & SFa & SFo & Sh & SFa & SFo & SFa & SFo & SFo\\
\hline\hline
\texttt{Kuopio} & 0.83 & 0.88 & 0.93 & 0.85 & 0.77 & 0.80 & 0.89 & 0.93 & 0.83 & 0.86 \\
\texttt{Venice} & 0.88 & 0.88 & 0.93 & 0.87 & 0.82 & 0.84 & 0.88 & 0.94 & 0.90 & 0.91 \\
\texttt{Bordeaux} & 0.76 & 0.87 & 0.95 & 0.81 & 0.72 & 0.75 & 0.87 & 0.90 & 0.79 & 0.81 \\
\texttt{Berlin} & 0.64 & 0.63 & 0.89 & 0.70 & 0.44 & 0.59 & 0.80 & 0.69 & 0.55 & 0.69 \\
\texttt{Rome} & 0.74 & 0.75 & 0.95 & 0.77 & 0.62 & 0.72 & 0.85 & 0.79 & 0.71 & 0.78 \\
\hline\hline
\end{tabular}
\caption{$\beta=300$ and Kendall, weighted Kendall, and intersection correlation}
\label{tbl:300pt}
\end{table*}

\begin{table*}
\begin{tabular}{||l||r|r|r|r|r|r|r|r|r|r||}
\cline{2-11}
\multicolumn{1}{c||}{} & \multicolumn{10}{|c||}{Kendall's $\tau$}\\
\cline{2-11}
\multicolumn{1}{c||}{} & \multicolumn{4}{c|}{Fa} & \multicolumn{3}{|c|}{Fo} & \multicolumn{2}{|c|}{Sh} & \multicolumn{1}{|c||}{SFa}\\
\hline
Network & Fo & Sh & SFa & SFo & Sh & SFa & SFo & SFa & SFo & SFo\\
\hline\hline
\texttt{Kuopio} & 0.80 & 0.82 & 0.93 & 0.83 & 0.71 & 0.76 & 0.85 & 0.87 & 0.78 & 0.83 \\
\texttt{Venice} & 0.85 & 0.82 & 0.91 & 0.84 & 0.75 & 0.81 & 0.85 & 0.88 & 0.84 & 0.88 \\
\texttt{Bordeaux} & 0.75 & 0.75 & 0.95 & 0.80 & 0.63 & 0.73 & 0.86 & 0.78 & 0.71 & 0.80 \\
\texttt{Berlin} & 0.64 & 0.57 & 0.89 & 0.70 & 0.40 & 0.59 & 0.78 & 0.63 & 0.52 & 0.70 \\
\texttt{Rome} & 0.75 & 0.64 & 0.95 & 0.76 & 0.55 & 0.73 & 0.85 & 0.67 & 0.64 & 0.77 \\
\hline\hline
\end{tabular}
\caption{$\beta=600$ and Kendall, weighted Kendall, and intersection correlation}
\end{table*}

\begin{table*}
\begin{tabular}{||l||r|r|r|r|r|r|r|r|r|r||}
\cline{2-11}
\multicolumn{1}{c||}{} & \multicolumn{10}{|c||}{Kendall's $\tau$}\\
\cline{2-11}
\multicolumn{1}{c||}{} & \multicolumn{4}{c|}{Fa} & \multicolumn{3}{|c|}{Fo} & \multicolumn{2}{|c|}{Sh} & \multicolumn{1}{|c||}{SFa}\\
\hline
Network & Fo & Sh & SFa & SFo & Sh & SFa & SFo & SFa & SFo & SFo\\
\hline\hline
\texttt{Kuopio} & 0.79 & 0.72 & 0.92 & 0.83 & 0.63 & 0.76 & 0.82 & 0.78 & 0.72 & 0.84 \\
\texttt{Venice} & 0.84 & 0.77 & 0.91 & 0.83 & 0.69 & 0.79 & 0.84 & 0.82 & 0.80 & 0.87 \\
\texttt{Bordeaux} & 0.78 & 0.66 & 0.95 & 0.80 & 0.59 & 0.76 & 0.86 & 0.69 & 0.68 & 0.82 \\
\texttt{Berlin} & 0.64 & 0.53 & 0.89 & 0.70 & 0.38 & 0.59 & 0.77 & 0.59 & 0.50 & 0.70 \\
\texttt{Rome} & 0.76 & 0.61 & 0.95 & 0.77 & 0.53 & 0.74 & 0.85 & 0.64 & 0.61 & 0.78 \\
\hline\hline
\end{tabular}
\caption{$\beta=1200$ and Kendall, weighted Kendall, and intersection correlation}
\end{table*}

\begin{table*}
\begin{tabular}{||l||r|r|r|r|r|r|r|r|r|r||}
\cline{2-11}
\multicolumn{1}{c||}{} & \multicolumn{10}{|c||}{Kendall's $\tau$}\\
\cline{2-11}
\multicolumn{1}{c||}{} & \multicolumn{4}{c|}{Fa} & \multicolumn{3}{|c|}{Fo} & \multicolumn{2}{|c|}{Sh} & \multicolumn{1}{|c||}{SFa}\\
\hline
Network & Fo & Sh & SFa & SFo & Sh & SFa & SFo & SFa & SFo & SFo\\
\hline\hline
\texttt{Kuopio} & 0.80 & 0.68 & 0.92 & 0.84 & 0.61 & 0.76 & 0.81 & 0.73 & 0.70 & 0.85 \\
\texttt{Venice} & 0.84 & 0.72 & 0.91 & 0.82 & 0.65 & 0.79 & 0.83 & 0.77 & 0.76 & 0.86 \\
\texttt{Bordeaux} & 0.79 & 0.65 & 0.95 & 0.82 & 0.60 & 0.78 & 0.86 & 0.69 & 0.69 & 0.83 \\
\texttt{Berlin} & 0.65 & 0.51 & 0.89 & 0.70 & 0.37 & 0.59 & 0.77 & 0.57 & 0.49 & 0.71 \\
\texttt{Rome} & 0.77 & 0.60 & 0.95 & 0.77 & 0.53 & 0.75 & 0.85 & 0.64 & 0.61 & 0.79 \\
\hline\hline
\end{tabular}
\caption{$\beta=2400$ and Kendall, weighted Kendall, and intersection correlation}
\end{table*}

\begin{table*}
\begin{tabular}{||l||r|r|r|r|r|r|r|r|r|r||}
\cline{2-11}
\multicolumn{1}{c||}{} & \multicolumn{10}{|c||}{Kendall's $\tau$}\\
\cline{2-11}
\multicolumn{1}{c||}{} & \multicolumn{4}{c|}{Fa} & \multicolumn{3}{|c|}{Fo} & \multicolumn{2}{|c|}{Sh} & \multicolumn{1}{|c||}{SFa}\\
\hline
Network & Fo & Sh & SFa & SFo & Sh & SFa & SFo & SFa & SFo & SFo\\
\hline\hline
\texttt{Kuopio} & 0.80 & 0.65 & 0.92 & 0.84 & 0.58 & 0.76 & 0.80 & 0.70 & 0.67 & 0.85 \\
\texttt{Venice} & 0.83 & 0.67 & 0.91 & 0.81 & 0.61 & 0.78 & 0.82 & 0.73 & 0.72 & 0.85 \\
\texttt{Bordeaux} & 0.80 & 0.64 & 0.95 & 0.81 & 0.59 & 0.78 & 0.85 & 0.67 & 0.68 & 0.83 \\
\texttt{Berlin} & 0.64 & 0.49 & 0.89 & 0.70 & 0.37 & 0.59 & 0.76 & 0.55 & 0.48 & 0.71 \\
\texttt{Rome} & 0.77 & 0.60 & 0.95 & 0.77 & 0.53 & 0.75 & 0.86 & 0.63 & 0.61 & 0.79 \\
\hline\hline
\end{tabular}
\caption{$\beta=\infty$ and Kendall, weighted Kendall, and intersection correlation}
\label{tbl:9999pt}
\end{table*}

\clearpage
\bibliographystyle{plainurl}
\bibliography{references.bib}

\end{document}